\documentclass[onecolumn]{IEEEtran}

\usepackage[utf8]{inputenc} 
\usepackage[T1]{fontenc}
\usepackage{url}
\usepackage{ifthen}
\usepackage{cite}
\usepackage[cmex10]{amsmath} 
\usepackage{upgreek}
\usepackage{comment}
\usepackage{mathrsfs, mathtools, float}
\usepackage{amsfonts,dsfont,amssymb,amsthm,stmaryrd,bbm,color}
\usepackage{hyperref,xcolor}
\usepackage{graphicx}
\usepackage{mathabx}

\newtheorem{conj}{Conjecture}[section]
\newtheorem{thm}{Theorem}[section]

\newtheorem{lem}[conj]{Lemma}
\newtheorem{prop}[conj]{Proposition}

\newtheorem{defn}[conj]{Definition}
\newtheorem{coro}[conj]{Corollary}

\newcommand{\bens}{\begin{equation*}}
\newcommand{\eens}{\end{equation*}}

\usepackage{authblk}

\title{The Differential Entropy of Mixtures: New Bounds and Applications} 
\author[1]{James Melbourne
}
\author[1]{Saurav Talukdar
}
\author[1]{Shreyas Bhaban
}
\author[2]{Mokshay Madiman
}
\author[1]{Murti Salapaka
}
\affil[1]{Electrical and Computer Engineering, University of Minnesota}
\affil[2]{Department of Mathematical Sciences, University of Delaware}

\begin{document}
\date{December 17, 2019}

\maketitle

\begin{abstract}
 Mixture distributions are extensively used as a modeling tool in diverse areas from machine learning to communications engineering to physics, and
 obtaining bounds on the entropy of probability distributions is of fundamental importance in many of these applications. 
 This article provides sharp bounds on the entropy concavity deficit, which is the difference between the entropy of the mixture and 
 the weighted sum of entropies of constituent components. Toward establishing lower and upper bounds on the concavity deficit, 
 results that are of importance in their own right are obtained. In order to obtain nontrivial upper bounds, properties of the 
 skew-divergence are developed and notions of ``skew'' $f$-divergences are introduced;
 a reverse Pinsker inequality and a bound on  Jensen-Shannon divergence  are obtained along the way.  
Complementary lower bounds are derived with special attention paid to the case that corresponds to independent summation of a continuous and a discrete random variable.  
Several applications of the bounds are delineated, including to  mutual information of additive noise channels, 
thermodynamics of computation, and functional inequalities.
\end{abstract}

\section{Introduction}
Mixture models are extensively employed in diverse disciplines including genetics, biology, medicine, economics, speech recognition, as the distribution of a signal at the 
receiver of a communication channel when the transmitter sends a random element of a codebook, or in models of clustering or classification in machine learning 
(see, e.g.,  \cite{schlattmann2009medical, grandvalet2005semi}). A mixture model is described by a density of the form
$f=\sum_i p_i f_i(x)$,  where each $f_i$ is a probability density function and each $p_i$ is a nonnegative weight with $\sum_i p_i= 1$. 
Such mixture densities have a natural probabilistic meaning as outcomes of a two stage random process with the first stage being a random draw, $i$, 
from the probability mass function $p$, followed by choosing a real-valued vector $x$ following a distribution $f_i(\cdot)$; equivalently,  it is the density of $X+Z$, 
where $X$ is a discrete random variable taking values $x_i$ with probabilities $p_i$, and $Z$ is a dependent variable such that $\mathbb{P}(Z \in A  | X= x_i) = \int_A f_i(z-x_i) dz$.  
The differential entropy of this mixture is of significant interest. 

Our original motivation for this paper came from the fundamental study of thermodynamics of computation, in which memory models are 
well approximated by mixture models, while the erasure of a bit of information  is akin to the state described by a single unimodal density. 
Of fundamental importance here is the entropy of the mixture model which is used to estimate the thermodynamic change in entropy 
in an erasure process and for other computations \cite{landauer1961irreversibility,talukdar2017memory}. It is not possible, analytically, to determine 
the differential entropy of the mixture model $\sum p_if_i$, even in the simplest case where $f_i$ are normal distributions, and hence one is interested in refined 
bounds on the same. While this was our original motivation, the results of this paper are more broadly applicable and we strive to give general statements
so as not to limit the applicability.

For a random vector $Z$ taking values in $\mathbb{R}^d$ with probability density density $f$, 
the {\it differential entropy} is defined as $h(Z) = h(f)= -\int_{\mathbb{R}^d} f(z) \ln f(z) dz$, where
the integral is taken with respect to Lebesgye measure. We will frequently
omit the qualifier ``differential'' when this is obvious from context and simply call it the entropy.
It is to be noted that, unlike $h(\sum p_if_i)$, 
the quantity $ \sum_i p_i h(f_i)$ is more readily determinable and thus the {\it concavity deficit} $h(\sum p_if_i)-\sum p_i h(f_i)$ is of interest. 
This quantity can also be interpreted as a generalization of the Jensen-Shannon divergence \cite{BH09}, and its quantum analog 
(with density functions replaced by density matrices and Shannon entropy replaced by von Neumann entropy) is the Holevo information, 
which plays a key role in Holevo's theorem bounding the amount of accessible (classical) information in a quantum state \cite{holevo2012quantum}. 

It is a classical fact going back to the origins of information theory that the entropy $h$ is a concave function, which means that the concavity deficit
is always nonnegative:  
\begin{align} \label{eq:conc-basic}
h(f)-\sum_i p_i h(f_i) \geq 0. 
\end{align}
Let  $X$ be a random variable that takes values in a countable set where $X=x_i$ with probability $p_i$. 
Intimately related to the entropy $h$ of a mixture distribution $f=\sum p_i f_i$ and the concavity deficit are the quantities $H(p):=-\sum p_i \log p_i$, and the conditional entropies $h(Z|X)$ and $H(X|Z)$. 
Indeed, it is easy to show  an upper bound on the concavity deficit (see, e.g., \cite{WM2014beyond}) in the form
\begin{align} \label{eq: discrete-continuous entropy inequality}
     h\left( f \right)-  \sum_i p_i h(f_i) \leq H(p) ,
\end{align} 
which relates the entropy of continuous variable $Z$ with density $f=\sum_i p_i f_i$ to the entropy of a random variable that lies in a countable set.  
A main thrust  of this article will be to provide refined upper and lower bounds on the concavity deficit that improve upon the basic bounds \eqref{eq:conc-basic} 
and \eqref{eq: discrete-continuous entropy inequality}. 

The main upper bound we establish is inspired by bounds in the quantum setting developed by Audenaert \cite{audenaert2014quantum} 
and utilizes the total variation distance. 
Given two two probability densities $f_1$ and $f_2$  with respect to a common measure $\mu$,
the total variation distance between them is defined  as $\|f_1-f_2\|_{TV}=\frac{1}{2} \int |f_1-f_2|d\mu.$

We will state the following theorem in terms of the usual differential entropy, on Euclidean space with respect to the Lebesgue measure. Within the article, the statements and proofs will be given for a general Polish measure space $(E, \gamma)$ from which the result below can be recovered as a special case.

\begin{thm} \label{thm: discrete continuous entropy inequality sharpening}
    Suppose $f = \sum_i p_i f_i$, where $f_i$ are probability density functions on $\mathbb{R}^d,$ $p_i\geq 0$, $\sum_ip_i=1.$ 
    Define the  mixture complement of $f_j$ by $\tilde{f}_j(z) = \sum_{i \neq j} \frac{p_i}{1-p_j} f_i$. Then
    \bens
       h(f)- \sum_i p_i h(f_i) \leq \mathcal{T}_f H(p) 
    \eens
    where 
    \bens
        \mathcal{T}_f \coloneqq \sup_i  \|f_i - \tilde{f}_i\|_{TV}.
    \eens
\end{thm}

Theorem~\ref{thm: discrete continuous entropy inequality sharpening} shows that as distributions cluster in total variation distance, the concavity deficit vanishes.   
The above result thus considerably reduces the conservativeness of the upper bound on the concavity deficit given by \eqref{eq: discrete-continuous entropy inequality}. 
Indeed consider the following example with  $f_1(z) = e^{-(z-a)^2/2}/\sqrt{2\pi}$ and $f_2(z) = e^{-(z+a)^2/2}/{\sqrt{2\pi}}$.  By \eqref{eq: discrete-continuous entropy inequality}, for $p \in (0,1)$,
\bens
     h(p f_1 + (1-p)f_2) \leq H(p) + \frac 1 2 \log 2 \pi e.
\eens
 However, noting that,  $\tilde{f}_1 = f_2$ and $\tilde{f}_2 = f_1$ implies,
 \bens\begin{split}
     \mathcal{T}_f 
        &= 
            \|f_1 - f_2\|_{TV}
                \\
        &=
            \int_0^\infty \left(e^{-(z-a)^2/2}/\sqrt{2 \pi} - e^{-(z+a)^2/2}/\sqrt{2 \pi} \right)dz
                \\
        &=
            \Phi(a) - \Phi(-a),
 \end{split}\eens
 where $\Phi$ is the standard normal distribution function
$
     \Phi(t) \coloneqq \int_{-\infty}^t e^{x^2/2}/\sqrt{2 \pi} dx.
$
Since $\Phi(a) - \Phi(-a) \leq a \sqrt{2/\pi}$, Theorem \ref{thm: discrete continuous entropy inequality sharpening} gives
\begin{align}
    h(p f_1 + (1-p)f_2) \leq a \sqrt{\frac{2}{\pi}}  H(p) + \frac 1 2 \log 2 \pi e.
\end{align}
 Another interpretation of Theorem \ref{thm: discrete continuous entropy inequality sharpening}, is as a generalization of the classical bounds on the Jensen-Shannon divergence by the total variation distance \cite{lin1991divergence, topsoe2003jenson}, which is recovered by taking $p_1 = p_2 = \frac 1 2$, see Corollary \ref{cor: Topsoe and Lin}.  Let us point out that $\Phi(a) - \Phi(-a) = 1 - \mathbb{P}(|\mathcal{Z}| > a)$ where $\mathcal{Z}$ is a standard normal variable, thus an alternative representation of these bounds in this special case is as mutual information bounds
 \begin{align} \label{eq: first mutual information bounds}
      I(X; \mathcal{Z}) \leq H(X) - \mathbb{P}(|\mathcal{Z}| > a ) H(X)
 \end{align}
 where $X$ denotes an independent Bernoulli taking the values $\pm a$ with probability $p$ and $1-p$.

 
 The methods and technical development  toward  establishing Theorem \ref{thm: discrete continuous entropy inequality sharpening} are of independent interest. 
We develop a notion of skew $f$-divergence for general $f$-divergences generalizing the skew divergence (or skew relative entropy) introduced by Lee \cite{lee1999measures}, 
 and in Theorem \ref{thm: skewing preserves f divergence} show that the class of $f$-divergences is stable under the skew operation.  
 After proving elementary properties of the skew relative entropy in Proposition \ref{prop: divergence skew properties} and 
 an introduced skew chi-squared divergence in Proposition \ref{prop: chi skew properties}, we adapt arguments due to Audenaert \cite{audenaert2014quantum} 
 from the quantum setting to prove the two $f$-divergences to be intertwined through a differential equality and that the 
 classical upper bound of the relative entropy by the chi-square divergence can be generalized to the skew setting 
 (see Theorem \ref{thm: skew chi squared related to skew relative entropy}).  After further bounding the skew chi-square divergence 
 by the total variation distance, we integrate the differential equality and obtain a bound of the skew divergence by the total variation 
 in Theorem \ref{thm: skews bounded by total variation}.  As a corollary we obtain a reverse Pinsker inequality due to Verdu \cite{verdu2014total}.  
 With these tools in hand, Theorem \ref{thm: discrete continuous entropy inequality sharpening} is proven and we demonstrate that the 
 bound of the Jensen-Shannon divergence by total variation distance \cite{lin1991divergence, topsoe2003jenson} is an immediate special case.

In the converse direction that provides lower bounds on the concavity deficit, our main result applies to the case where 
all the component densities come from perturbations of a random vector $W$ in $\mathbb{R}^d$ that has a log-concave and spherically symmetric distribution. 
 We say a random vector $W$ has a  {\it log-concave} distribution when it possesses a density $\varphi$ satisfying $\varphi((1-t)z+ty) \geq \varphi^{1-t}(z) \varphi^t(y)$.
 We say $W$ has a  {\it spherically symmetric} distribution when there exists $\psi:\mathbb{R}\rightarrow [0,\infty)$ such that
 the density $\varphi(z)=\psi(|z|)$ for every $z\in \mathbb{R}^d$, where $|z| \coloneqq \sqrt{z_1^2+z_2^2+\ldots z_d^2}$. 
 We employ the notation $B_\lambda = \{x\in \mathbb{R}^d : |x| \leq \lambda \}$ for the centered closed ball of radius $\lambda$ in $\mathbb{R}^d$, 
$\mathscr{T}(t) \coloneqq \mathscr{T}_W(t) \coloneqq \mathbb{P}(|W| > t)$ for the tail probability of $W$, and 
$\|A\|:=\sup_{\|w\|_2=1}\|Aw\|_2$ for the operator norm of a matrix $A:\mathbb{R}^d\rightarrow \mathbb{R}^d$.
\begin{thm} \label{thm: Fano Sharpening simple}
    Suppose that there exists $\tau \geq 1$ such that for each $x \in \mathcal{X}$, $Z|X=x$ has distribution given by $T_x(W)$ 
    where $W$ has density $\varphi$, spherically symmetric and log-concave and $T_x$ is a $\sqrt{\tau}$ bi-Lipschitz function.  
    For $i,k \in \mathcal{X}$, take $T_{ij} \coloneqq T_i^{-1} \circ T_j$ and further assume there exists $\lambda >0$ such that for 
    any $k \neq i$ $ T_{ij}(B_\lambda) \cap  T_{kj}(B_\lambda) = \emptyset$.  Then 
    \begin{align}
        h(Z) -  h(Z|X) \geq H(X) - \tilde{C}(W),
    \end{align}
    where $\tilde{C}$ is the following function dependent heavily on the tail behavior of $|W|$,
    \begin{align}
        \tilde{C}(W) =  \mathscr{T}(\lambda)(1 + h(W))  +  \mathscr{T}^{\frac 1 2}(\lambda) (\sqrt{d} + K(\varphi))
    \end{align}

with
\begin{align} \label{eq: definition of K(varphi)}
    K(\varphi) \coloneqq 
        \log \left[ \tau^d  \left( \|\varphi\|_\infty + \left(\frac{3}{\lambda} \right) \omega_d^{-1} \right) \right] \mathbb{P}^{\frac 1 2}(|W| > \lambda) + d \left( \int_{B_\lambda^c} \varphi(w)\log^{2}\left[ 1 + \tau + \frac{\tau^2 |w|}{ \lambda}\right] dw \right)^{\frac 1 2}
\end{align}
where $\omega_d$ denoting the volume of the $d$-dimensional unit ball, $B_\lambda^c$ denotes the complement of $B_\lambda \in \mathbb{R}^d$.  
\end{thm}

We note the quantity $H(X|Z)$ connotes the uncertainty in the discrete variable $X$ conditioned on the continuous variable $Z;$ such a quantity needs to be defined/determined from the knowledge of probabilities, $p_i,$  that the discrete variable $X=x_i$ and the description of the conditional probability density function $p(z|x_i)=f_i(z).$ These notions are made precise in Section~\ref{sec: Notation}. Here, it is also established that   \eqref{eq: discrete-continuous entropy inequality} can be equivalently formulated as $H(X|Z) \leq H(X)$ for a particular coupling of a discrete variable $X$ taking values with probabilities $\{p_i\}$, and a variable $Z$ with density $f_i$ when conditioned on  $X = x_i$.  From this perspective the super-concavity bound of Theorem \ref{thm: Fano Sharpening} gives $H(X|Z) \leq \tilde{C}(W)$.  One should also note that when $\{p_i\}_{i=1}^n$ is a finite sequence, the classical bounds on $H(X|Z)$ are provided by Fano's inequality:  for a Markov triple of random variables $X \to Z \to \hat{X}$, and $e =\{ X \neq \hat{X} \}$,
\begin{align}\label{eq: The FAno's inequality}
    H(X|Z) \leq H( e) + \mathbb{P}(e) \log( \#\mathcal{X} - 1),
\end{align}
where $\#\mathcal{X}$ denotes the cardinality of the set $\mathcal{X},$
gives a strengthening of concavity in many situations, where we have employed the notation for a measurable set $A$, $H(A) =  - \mathbb{P}(A) \log \mathbb{P}(A) - (1-\mathbb{P}(A)) \log (1- \mathbb{P}(A))$.  It yields
\begin{align}
    h(f) \geq  \sum_{i=1}^n p_i h(f_i) +  H(p) - \left(H( e) + \mathbb{P}(e) \log( \#\mathcal{X} - 1) \right).
\end{align}
To compare the strength of the bounds derived in Theorem \ref{thm: Fano Sharpening} to Fano's we compare $H( e) + \mathbb{P}(e) \log( \#\mathcal{X} - 1)$ and $\tilde{C}(W)$; as is established in Section~\ref{sec: Upper bounds}, even in simple cases, $\tilde{C}(W)$ can be arbitrarily small even while $\min_{\hat{X}} H(e) + \mathbb{P}(e) \log( \#\mathcal{X}-1)$ is arbitrarily large.  \\

The study of entropy of mixtures has a long history and is scattered in a variety of papers that often have
other primary emphases. Consequently it is difficult to exhaustively review all related work. Nonetheless, the references that we were able
to find that attempt to obtain refined bounds under various circumstances are \cite{HBDH08, AB11:isit, BM11:it, KHP11, MK16, NN17}.
In all of these papers, however, either the bounds deal with specialized situations, or with a general setup but employing different
(and typically far more) information than we require. We emphasize that our bounds deal with general multidimensional situations
(including in particular multivariate Gaussian mixtures, which are historically and practically of high interest) 
and in that sense go beyond the previous literature.

The article is organized as follows.    In Section \ref{sec: Notation} we will give notation and preliminaries, where we delineate definitions and relationship for  entropies, conditional entropies emphasizing  a mix of continuous and discrete variables.   Section \ref{sec: Lower Bounds} is devoted to the proof of Theorem \ref{thm: discrete continuous entropy inequality sharpening} (a preliminary verison has  appeared in the conference paper \cite{melbourne2019relationships}).  In Section \ref{sec: Upper bounds} we prove Theorem \ref{thm: Fano Sharpening}.  These results give a considerable generalization of earlier work of authors in \cite{melbourne2018isitLandauer, talukdar2018analysis}. The result will hinge on a Lemma \ref{lem: sums bounded at well spaced points} bounding the sum 
$\sum_i \varphi(x_i)$ for $x_i$ well spaced and $\varphi$ log-concave and spherically symmetric, and a concentration result from convex geometry \cite{FMW16}.  
We close discussing bounds of $\mathbb{P}(|W| \geq t)$ in the case that $W$ is log-concave and strongly log-concave, see Corollary \ref{cor: tail bounds for log-concave and strongly}.
Section \ref{sec: applications} we demonstrate applications of the theorems to a diverse group of problems; hypothesis testing, capacity estimation, nanoscale energetics, and functional inequalities.

\section{Notation and Preliminaries} \label{sec: Notation}
In this part of the article, we will elucidate definitions and results for conditional entropy and mutual information when a mix of discrete valued and continuous random variables are involved.  We will assume that 
\begin{enumerate}
    \item 
    $X$ takes values in a discrete countable set $\mathcal{X}$, and that $\mathbb{P}(X = x) = p_x >0$.
    \item
    $Z$ is a random variable which takes values in a Polish space $E$.  The conditional distribution is described by $\mathbb{P}(Z \in A| X=x) = \int_A f_x(z) d\gamma(x)$ where $f_x(z)$ is a density function with respect to a Radon reference measure $\gamma$.
\end{enumerate}
We will denote by $m$ the counting measure on $\mathcal{X}$, so that for $A \subseteq \mathcal{X}$, the measure of $A$ is  its cardinality, 
\begin{align}
    m(A) = \#(A).
\end{align}  
Integration with respect to the counting measure, corresponding to summation; for $g: \mathcal{X} \to \mathbb{R}$ such that $\sum_{x \in \mathcal{X}} |g(x)| < \infty$, 
\begin{align}
    \int_\mathcal{X} g(x) dm  := \sum_{x \in \mathcal{X}} g(x).
\end{align}

We will denote by $dm \hspace{1mm} d\gamma$ the product measure on $\mathcal{X} \times E$ where,  for $A \subseteq \mathcal{X}$ and measurable $B \subseteq \mathbb{R}^d$,
\begin{align}
    \int_{\mathcal{X} \times E} \mathbbm{1}_{A \times B}(x,z) dm(x) \hspace{1mm} d\gamma(z)
        \coloneqq m(A) \gamma(B),
\end{align}
when $\gamma$ denotes the $d$-dimensional Lebesgue measure, we will use $|B|_d$ or $|B|$ when there is no risk of confusion to denote the Lebesgue volume of a measureable set $B$.  For measures $\mathbb{P}$ and $\mathbb{Q}$ on a shared measure space such that $\mathbb{P}$ has a density $\varphi$, with respect to $\mathbb{Q}$,  when any measurable set $A$ satisfies,
\begin{align}
    \mathbb{P}(A) = \int_A \varphi \hspace{1mm} d \mathbb{Q}.
\end{align}
Such a $\varphi$ will also be written as $\frac{d \mathbb{P}}{d\mathbb{Q}}$.

For random variables $U$ and $V$ whose induced probability measures admit densities with respect to a reference measure $\gamma$, in the sense that $\mu(A) \coloneqq \mathbb{P}(U \in A) = \int_A u d\gamma$ and $\nu(A) \coloneqq \mathbb{P}(V \in A) = \int_A v d\gamma$ where $u$ and $v$ are density functions with respect to $\gamma$, the relative entropy (or KL divergence) is defined as
\begin{align}
    D(U || V ) \coloneqq D(\mu || \nu) \coloneqq D(u || v) \coloneqq \int u \log \frac {u}{v} d \gamma.
\end{align}
\noindent
When $U$ has an $E$ valued random variable density $u$ with respect to a reference measure $\gamma$, denote the entropy
\begin{align}
    h_\gamma(U) \coloneqq  h_\gamma(u) = - \int u(z) \log u(z) d \gamma(z),
\end{align}
whenever the above integral is well defined.
When $\gamma$ is the Lebesgue measure we denote the usual differential entropy,
\begin{align}
    h(U) \coloneqq  h(u) \coloneqq- \int u(x) \log u(x) dx,
\end{align}

When $U$ is discrete, taking values $x \subseteq{X}$ with probability $p_x$, define
\begin{align}
    H(U) \coloneqq H(p) \coloneqq - \sum_{x \in \mathcal{X}} p_x \log p_x.
\end{align}
When $t \in [0,1]$, we define $H(t)$ to the the entropy of a Bernoulli random variable with parameter $t$, $H(t) \coloneqq -(1-t) \log (1-t) - t \log t$.  When $A$ is an event, $H(A) \coloneqq H(\mathbb{P}(A))$.

The following proposition elucidates  notions of conditional entropy and joint entropy of a mix of discrete and continuous random variables.
\begin{prop} \label{prop: joint distribution of X and Z}
Suppose $X$ is a discrete random variable with values in a countable set $\mathcal{X}$ and $Z$ is a Borel measurable random variable taking values in $E$. Suppose, for all $x \in \mathcal{X}$, $P(Z\in A|X=x)=\int_A f_x (z) d\gamma(z)$ for density $f_x(z)$ with respect to a common reference measure $\gamma$. Then the following hold.
\begin{itemize}
    \item The joint distribution of $(X,Z)$ on $\mathcal{X} \times E$, has a density 
\begin{align} \label{eq: Joint density function of XZ}
    F(x,z) = p_x f_x(z)
\end{align}
with respect to $dm \hspace{1mm} d\gamma$. 

\item  $Z$ has a density 
\begin{align} \label{eq: density of Z}
    f(z) = \sum_{x \in \mathcal{X}} p_x f_x(z)
\end{align} with respect to $\gamma$ on $E$. 

\item The conditional density of $X$ with respect to $Z=z$ defined as  \begin{align}
    p(x|z) = \begin{cases} \frac{p_x f_x(z)}{f(z)} &\mbox{ for } f(z) >0 \\ 
    0 &\mbox{ otherwise},
    \end{cases}
\end{align} satisfies $\mathbb{P}(X=x) = \int_{E} p(x|z) f(z) dz = p_x$.
\end{itemize}

\end{prop}

\begin{proof}
    Note that since a set $A \subseteq \mathcal{X} \times E$ can be decomposed into a countable union of disjoint sets $\{x\} \times A_x$ where $A_x = \{ z \in E : (x,z) \in A\}$, to prove $F(x,z)$ is the joint density function of $(X,Z)$, it suffices to prove 
    \begin{align} \label{eq: equate the summands}
        \mathbb{P}_{XZ}(\{x\}\times A) = \int_{ \{x\} \times A} F(x,z) dm(x) \hspace{1mm} d \gamma(z).
    \end{align} Indeed,
    \begin{align}
        \mathbb{P}_{XZ}(A) 
            &=
                \mathbb{P}_{XZ}(\cup_x \{x\}\times A_x )
                    \\
            &=
                \sum_x \mathbb{P}_{XZ} (\{x\} \times A_x ), \label{eq: sum in terms of measures}
    \end{align}
    while
    \begin{align}
        \int_A F(x,z) dm(x) \hspace{1mm} d\gamma(z)
            &=
                \int_{ \cup_x \{x\} \times A_x } F(x,z) dm(x) \hspace{1mm} d\gamma(z)
                    \\
            &=
                \sum_x \int_{ \{x \} \times A_x } F(x,z) dm(x) \hspace{1mm} d \gamma(z).\label{eq: sum in terms of integrals}
    \end{align}
    Since \eqref{eq: equate the summands} would give equality of the summands of \eqref{eq: sum in terms of measures} and \eqref{eq: sum in terms of integrals} the result would follow.
    We compute directly.
    \begin{align}
        \mathbb{P}_{XZ}(\{x\} \times A) 
            &= 
                \mathbb{P}(X=x , Z \in A)
                    \\
            &=
                \mathbb{P}(X=x) \mathbb{P}(Z \in A | X =x)
                    \\
            &=
                p_x \int_A f_x(z) d \gamma(z)
                    \\
            &=
                \int_{\{x\}} p_x \int_A  f_x(z) d\gamma(z) dm(x)
                    \\
            &=
                \int_{\{x\} \times A} F(x,z) dm(x) \hspace{1mm} d\gamma(z) .
    \end{align}
    This gives the first claim.  For the second,
    \begin{align}
        \mathbb{P}(Z \in A)
            &=
                \sum_{x \in \mathcal{X}} \mathbb{P}( X = x, Z \in A)
                    \\
            &=
                \sum_{x \in \mathcal{X}} p_x \int_A f_x(z) d \gamma(z)
                    \\
            &=
                \int_A f(z) d \gamma(z).
    \end{align}
    The last assertion is immediate,
    \begin{align}
        \int_{E} p(z|x) f(z) d \gamma(z)
            =
                \int_{E} \frac{p_x f_x(z)}{f(z)} f(z) d \gamma(z) = p_x.
    \end{align}
\end{proof}

\noindent Proposition \ref{prop: joint distribution of X and Z} allows the following definitions of conditional entropies.
\begin{itemize}
\item $h_\gamma(Z|X=x)=-\int_E f(Z|X=x)\log f(Z|X=x)=-\int f_x \log f_x d \gamma(z)$ and thus

\begin{equation} \label{eq: conditional of continuous with respect to discrete entropy}
    h(Z|X) \coloneqq E_x[h(Z|X=x)]=-\sum_{x\in \mathcal{X}}p_x\int_{z\in E} f_x(z) \log f_x(z) d \gamma(z)=  \sum_{x \in \mathcal{X}} p_x h_\gamma(f_x).
\end{equation}

\item  $H(X|Z=z)=-\sum_{x \in \mathcal{X}} p(X=x|Z=z)\log p(X=x|Z=z) =-\sum_{x \in \mathcal{X}} p(x|z)\log p(x|z)=H(p(\cdot|z)) $ and thus 
\begin{equation} 
    H(X|Z)=E_Z[H(X|Z=z)] \coloneqq -\int_E \left(\sum_{x \in \mathcal{X}} p(x|z)\log p(x|z)\right)  f(z) d\gamma(z).
\end{equation}
\end{itemize}
Let us note how the entropy of a mixture can be related to its relative entropy with respect to a dominating distribution $g$.  The entropy concavity deficit of a convex combination of densities $f_i$, is the convexity deficit of the relative entropy with respect to a reference measure in the following sense.

\begin{prop}\label{prop: compensation identity}
    For a density $g$ such that $\sum_x p_x D(f_x||g) < \infty$,
    \begin{align} \label{eq: entropyRelEntropy}
        h_\gamma(f) - \sum_x p_x h_\gamma(f_x) = \sum_x p_x D(f_x || g) - D( f || g).
    \end{align}
\end{prop}

\begin{proof}
    \begin{align*}
        \sum_x p_x D(f_x || g) - D( f || g) 
            &=
                \sum_x p_x \left(\int f_x \log \frac{f_x}{g} - f_x \log \frac f g d\gamma \right) 
                    \\
            &=    
                \sum_x p_x \left(\int f_x \log f_x - f_x \log f d\gamma \right)
                    \\
            &=
                h_\gamma(f) - \sum_x p_x h_\gamma(f_x).
    \end{align*}
\end{proof}
Note that the left hand side of Proposition \ref{prop: compensation identity} is invariant with respect to $g$.  Thus for $g_1, g_2$ such that  $\sum_x p_x D(f_x||g_j)  < \infty$,
\begin{align} 
    \sum_x p_x D(f_x || g_1) - D( f || g_1)  = \sum_x p_x D(f_x || g_2) - D( f || g_2).
\end{align}
Taking $g_1 = g$ and $g_2 = f = \sum_x p_x f_x$ yields the compensation identity,
\begin{align}
    \sum_x p_x D(f_x||g)  = \sum_x p_x D(f_x ||f) + D(f||g),
\end{align}
which is often used to obtain its immediate corollary
\begin{align}
    \min_g \sum_x p_x D(f_x||g) = \sum_x p_x D(f_x || f).
\end{align}

We define the mutual information between probability measures $\mathbb{P}_U$ and $\mathbb{P}_V$ with joint distribution $\mathbb{P}_{U V}$ and their product distribution $\mathbb{P}_U \mathbb{P}_V$, as the relative entropy of the product distribution from the joint distribution,
\begin{align*}
    I(\mathbb{P}_{U}; \mathbb{P}_V) = D( \mathbb{P}_{U V} || \mathbb{P}_U \mathbb{P}_V).
\end{align*}
For the random variables $U$ and $V$ inducing probability measures $\mathbb{P}_U$ and $\mathbb{P}_V$, we will write $I(U;V) = I(\mathbb{P}_U; \mathbb{P}_V)$.

\begin{prop} \label{prop: mutual information equality}
For $X$ discrete with $\mathbb{P}(X=x) = p_x$ and $Z$ satisfying $\mathbb{P}(Z \in B | X = x) =  \int_B f_x(z) dz$,
    \begin{align}
         I(X;Z) 
            &= 
                h_\gamma(Z) - h_\gamma(Z|X) \label{eq: continuous mutual information expression}
                    \\
            &=
                H(X) - H(X|Z) \label{eq: discrete mutual information expression}
                    \\
            &=
                \sum_{x \in \mathcal{X}} p_x D(f_x || f). \label{eq: sum of KL divergence expression}
    \end{align}
\end{prop}

\begin{proof}
By Proposition \ref{prop: joint distribution of X and Z}, $\mathbb{P}_{XZ}$ has density $F(x,z) = p_x f_x(z)$ with respect to $dm(x) \hspace{1mm} d\gamma(z)$ the product of the counting measure $m$ and $\gamma$.  The product measure $\mathbb{P}_X \mathbb{P}_Z$, has density $G(x,z) = p_x f(z)$ with respect to $dm \hspace{1mm} d\gamma$ and it follows that
\begin{align}
    \frac{ d \mathbb{P}_{XZ} }{ d \mathbb{P}_X \mathbb{P}_Z}(x,z)
        =
            \frac{ \frac{ d \mathbb{P}_{XZ} }{ dm \hspace{1mm} d\gamma}(x,z)}{\frac{ d \mathbb{P}_{X} \mathbb{P}_{Z} }{ dm \hspace{1mm} dz}(x,z)}
        =
            \frac{ F(x,z)}{G(x,z)} = \frac{f_x(z)}{f(z)} \label{eq: formula for radon nikodym derivative}
\end{align}
By equation \eqref{eq: formula for radon nikodym derivative},
\begin{align}
    D(\mathbb{P}_{XZ} || \mathbb{P}_X \mathbb{P}_Z)
        &=
            \int_{\mathcal{X} \times \mathbb{R}^d} F(x,z) \log \frac{f_x(z)}{f(z) } dm \hspace{1mm} d\gamma
                \\
        &=
            \int_{E} \sum_{x \in \mathcal{X}} p_x f_x(z) \log \frac{f_x(z)}{f(z) } d\gamma(z)
\end{align}
Recalling $p(z|x)$ from Proposition \ref{prop: joint distribution of X and Z}, using the algebra of logarithms and Fubini-Tonelli,
\begin{align}
    \int_{E} \sum_{x \in \mathcal{X}}& p_x f_x(z) \log \frac{f_x(z)}{f(z) } d \gamma(z)
        \\
        &=
            \int_{E} \sum_{x \in \mathcal{X}} p_x f_x(z) \log \frac{p(x|z)}{p_x } d \gamma(z)
                \\
        &=
            -\sum_{x \in \mathcal{X}} p_x \log p_x \int_{E} f_x(z)d \gamma(z) + \int_{E} f(z) \sum_{x \in \mathcal{X}} p(x|z) \log p(x|z) d\gamma(z)
                \\
        &=
            H(p) - \int_{\mathbb{R}^d} f(z) H(p(x|z))d \gamma(z)
                \\
        &=
            H(X) - H(X|Z),
\end{align}
giving \eqref{eq: discrete mutual information expression}.
By Fubini-Tonelli, 
\begin{align}
    \int_{E} \sum_{x \in \mathcal{X}} p_x f_x(z) \log \frac{f_x(z)}{f(z) } d\gamma(z)
        &=
            \sum_{x \in \mathcal{X}} p_x \int_{E} f_x(z) \log \frac{f_x(z)}{f(z)} d \gamma(z)
                \\
        &= 
            \sum_{x \in \mathcal{X}} p_x D(f_x||f),
\end{align}
we have expression \eqref{eq: sum of KL divergence expression}.  By Proposition \ref{prop: compensation identity},
\begin{align}
    \sum_{x \in \mathcal{X}} p_x D(f_x||f)
        &=
            h_\gamma(f) - \sum_{x \in \mathcal{X}} p_x h(f_x)
                \\
        &=
            h(Z) - h(Z|X),
\end{align}
\eqref{eq: continuous mutual information expression} follows.
\end{proof}

Using Proposition \ref{prop: mutual information equality}, we can give a simple information theoretic proof of a result proved analytically in \cite{WM2014beyond, bobkovmarsiglietti2019}.
\begin{coro}\label{cor:wm-mix}
    When $\mathcal{X} \subseteq E$, and $\gamma$ is a Haar measure, then $X$ and $Z$ satisfy,
    \begin{equation} \label{eq: discrete continuous entropy inequality}
        h_\gamma(X+Z) \leq H(X) + h_\gamma(Z|X)
    \end{equation}
    which reduces to 
    \begin{equation} \label{eq: discrete continuous entropy inequality independent}
        h_\gamma(X+Z) \leq H(X) + h_\gamma(Z)
    \end{equation}
    in the case that $X$ and $Z$ are independent.
\end{coro}

\begin{proof}
    Applying Proposition \ref{prop: mutual information equality} to $X$ and $\tilde{Z} = X+Z$ we have
    \begin{align}
        h_\gamma(X+Z) 
            &= 
                H(X) + h_\gamma(X+Z| X) - H(X|X+Z)
                    \\
            &=
                H(X) + h_\gamma(Z|X) - H(X|X+Z),
    \end{align}
    where the second equality follows from the assumption that $\gamma$ is a Haar measure.
    Since $H(X|X+Z) \geq 0$, \eqref{eq: discrete continuous entropy inequality} follows, while \eqref{eq: discrete continuous entropy inequality independent} follows from $h_\gamma(X+Z|X) = h_\gamma(Z)$ under the assumption of independence.
\end{proof}

Incidentally, the main use of Corollary~\ref{cor:wm-mix} in \cite{WM2014beyond} is to give a rearrangement-based proof of the entropy power inequality
(see \cite{MMX17:0} for much more in this vein).

\section{Lower bounds} \label{sec: Lower Bounds}
In this section we will provide lower bounds to the concavity deficit and provide  a proof of Theorem~\ref{thm: discrete continuous entropy inequality sharpening}.
We will first introduce the notion of $f$-divergences.


\begin{defn} \label{def: f-divergence}
    For a convex function $f$ satisfying $f(1) = 0$, and probability measures $\mu$ and $ \nu$, with densities $u = \frac{d\mu}{d\gamma}$ and $v = \frac{d \nu}{d\gamma}$ with respect to a common reference measure $\gamma$,
    the $f$ divergence from $\mu$ to $\nu$ is
    \begin{align} \label{eq: f -divergence definition}
        D_f( \mu || \nu) \coloneqq D_f(u||v) \coloneqq \int f\left( \frac{u}{v} \right) d v d\gamma.
    \end{align}
\end{defn}
Note that a common reference measure for measures $
\mu$ and $\nu$  always exists, take $\frac 1 2 (\mu + \nu)$ for instance, and the value of $D_f( \mu || \nu)$ is independent of the choice of reference measure as can be seen by comparing a reference measures to one it has a density with respect to.
When the value of a functional of a pair of probability distributions $\mu,\nu$ is given by \eqref{eq: f -divergence definition} we will call the functional an $f$-divergence.  An $f$-divergence satisfies the following.
(i) 
    Non-negativity, $D_f( \mu || \nu ) \geq 0$ and 
(ii) 
    The map $(\mu,\nu) \mapsto D_f( \mu || \nu)$ is convex.
We direct the reader to \cite{liese2006divergences,sason2018f,sason2016f} for further background on $f$-divergences and their properties. When $f(x) = x \log x$, the divergence induced is the relative entropy.   

This section is organized as follows. We  first  introduce the concept of skewing which is the  $f$-divergence from a convex combination $(1-t) \mu + t \nu$ to $\nu$. Skewing provides a more regular version of the original divergence measure, for example the Radon-Nikodym derivative  of $\mu$ with respect to $(1-t) \mu + t\nu$ always exists even if Radon-Nikodym derivative of $\mu$ with respect to $\nu$ may not, whereby skew divergence is well defined unlike divergence.  Skew divergence as we will develop preserves important features of the original divergence.  We first state elementary properties of the skew relative information, corresponding to skewing the relative entropy with proofs given in an appendix, and then introduce a skew $\chi^2$-divergence which interpolates between the well known Neyman $\chi^2$ divergence and the Pearson $\chi^2$ divergence.  

We will pause to demonstrate that the class of $f$-divergences is stable under skewing and recover as a special case; a recent result of Nielsen \cite{nielsen2019generalization}, that the generalized Jensen-Shannon divergence is an $f$-divergence. Then we establish several inequalities between the skew relative information and the introduced skew $\chi^2$ divergence.  We will show in Theorem \ref{thm: skew chi squared related to skew relative entropy} that the skew relative information can be controlled by the skew $\chi$-square divergence extending the classical bound of relative entropy by Pearson $\chi^2$ divergence, and using an argument due to Audenart in the quantum setting \cite{audenaert2014quantum}, we show that the rate of decrease of the skew relative information with respect to the skewing parameter can be described exactly as a multiple of the skew $\chi^2$ divergence.

Theorem \ref{thm: skews bounded by total variation} also appropriates a quantum argument \cite{audenaert2014quantum} to show that though neither the Neyman or Pearson divergences can be controlled by total variation, their skewed counterparts can be.  We harness this bound along side the  differential relationship between the two skew divergences to bound the skew relative entropy by the total variation as well.  As a brief aside we demonstrate that the this bound is equivalent to a reverse Pinsker type inequality derived by Verdu \cite{verdu2014total}, before using Theorem \ref{thm: skews bounded by total variation} to give our proof of Theorem \ref{thm: discrete continuous entropy inequality sharpening}. Finally to conclude the section, we demonstrate that one may obtain the classical result of Lin \cite{lin1991divergence} bounding the Jensen-Shannon divergence by total variation as a special case of Theorem \ref{thm: discrete continuous entropy inequality sharpening}.

\subsection{Skew Relative Information} \label{sec: Skew}
We will consider the following generalization of the relative entropy due to Lee.
\normalfont{
\begin{defn}\cite{lee1999measures}
For probability measures $\mu$ and $\nu$ on a common set $\mathcal{Y}$ and $t \in [0,1]$ define their Skew relative information 
\[
    S_t(\mu || \nu) = \int \log \frac{d \mu}{d(t \mu + (1-t) \nu)} d \mu
\]
In the case that $d \mu = u d\gamma$, and $d \nu = v \gamma$ we will also write
\[
    S_t(u||v) = S_t(\mu || \nu).
\]
\end{defn}
}

We state some important properties of Skew relative information with the proofs provided in the Appendix.
\begin{prop}\label{prop: divergence skew properties}
For probability measures $\mu$ and $\nu$ on a common set and $t \in [0,1]$ the Skew Relative information satisfies the following properties.
\begin{enumerate}
    \item \label{item: skew divergence to KL divergence}
    $S_t(\mu || \nu) = D( \mu || t \mu + (1-t) \nu)$.  In particular, $S_0(\mu || \nu) = D(\mu || \nu)$.
    \item \label{item: D one is zero} 
    $S_t(\mu || \nu) = 0$ iff $t=1$ or $\mu = \nu$. 
    \item \label{item: skew divergence bounded by log t}
    For $0<t <1$ the Radon-Nikodym derivative of $\mu$ with respect to $t\mu+(1-t)\nu$ does exist, and $S_t(\mu || \nu) \leq - \log t$.
    \item \label{item: function of t}
    $S_t(\mu || \nu)$ is convex, non-negative, and decreasing in $t$.
    \item \label{item: skew relative entropy is an f-divergence}
    $S_t$ is an $f$-divergence with $f(x) = x \log (x/(tx + (1-t)).$
\end{enumerate}
\end{prop}

Motivated by the fact that the act of {\it skewing} the relative entropy preserves its status as an $f$-divergence we introduce the act of skewing of an $f-divergence$

\begin{defn}
    Given a convex function $f: [0,\infty) \to \mathbb{R}$ with $f(1) = 0$ and its associated divergence $D_f(\cdot || \cdot )$, define the $r,t$-skew of $D_f$ by
    \begin{align} \label{eq: skew divergence defintion}
        S_{f,r,t}(\mu || \nu) \coloneqq D_f( r \mu + (1-r) \nu || t \mu + (1-t) \nu).
    \end{align}
    
        It can be shown that for $t\in (0,1)$,  $S_{f,r,t}(\mu || \nu) <\infty.$
\end{defn}

\begin{thm}  \label{thm: skewing preserves f divergence}
    The class of $f$-divergences is stable under skewing.  That is, if $f$ is convex, satisfying $f(1) =0$, then
    \begin{align}
        \hat{f}(x) \coloneqq (tx + (1-t))f\left(\frac {rx +(1-r)} {tx + (1-t)} \right)
    \end{align}
    is convex with $\hat{f}(1)=0$ as well, so that the $r,t$ skew of $D_f$ defined in \eqref{eq: skew divergence defintion} is an $f$-divergence as well. 
\end{thm}

\begin{proof}
    If $\mu$ and $\nu$ have respective densities $u$ and $v$ with respect to a reference measure $\gamma$, then $r \mu + (1-r) \nu$ and $t \mu + 1-t \nu$ have densities $r u + (1-r) v$ and $t u + (1-t)v$
    \begin{align}
        S_{f,r,t}( \mu || \nu) 
            &=
                \int f\left(\frac{ r u +(1-r)v}{t u + (1-t) v} \right) (t u + (1-t) v) d\gamma
                    \\
            &=
                \int f \left( \frac{r \frac u v + (1-r) }{t \frac u v + (1-t)} \right) (t \frac u v + (1-t)) v d \gamma
                    \\
            &=
                \int \hat{f}\left( \frac u v \right) v d\gamma.
    \end{align}
    Since $\hat{f}(1) = f(1) = 0$, we need only prove $\hat{f}$ convex.  For this, recall that the conic transform $g$ of a convex function $f$ defined by $g(x,y) = yf(x/y)$ for $y >0$ is convex, since 
    \begin{align}
        \frac{y_1+y_2}{2} f\left(\frac{x_1 + x_2 }{2}/ \frac{y_1 + y_2}{2} \right) 
            &= 
                \frac{y_1+y_2}{2} f\left(\frac{y_1}{y_1+y_2} \frac{x_1}{y_1} + \frac{y_2}{y_1+y_2} \frac{x_2}{y_2} \right) 
                    \\
            &\leq \frac{y_1}{2} f(x_1/y_1) + \frac{y_2}{2} f(x_2/y_2).
    \end{align}
   Our result follows since $\hat{f}$ is the composition of the affine function $A(x) = (rx + (1-r), tx + (1-t))$ with the conic transform of $f$,
    \begin{align}
        \hat{f}(x) = g(A(x)).
    \end{align}
\end{proof}

  Let us note that in the special case that $D_f$ corresponds to relative entropy, Theorem \ref{thm: skewing preserves f divergence} demonstrates that the ``Generalized Jensen-Shannon divergence'' developed recently by  Nielsen see \cite[Definition 1]{nielsen2019generalization} is in fact an $f$-divergence, as it is defined as the weighted sum of $r_i,t$-skew divergences associated to the relative entropy.

\begin{coro}
For a vector $\alpha \in [0, 1]^k$ and $w_i>0$ such that $\sum_i w_i = 1$, the $(\alpha, w)$-Jensen-Shannon divergence between two densities $p,q$
defined by:
\begin{align}
    JS^{\alpha,w}(p : q) \coloneqq \sum_{i=1}^k w_i D( (1-\alpha_i) p + \alpha_i q|| (1-\bar{\alpha}) p + \bar{\alpha} q)
\end{align}
with $\bar{\alpha} = \sum_i w_i \alpha_i$, is an $f$-divergence.
\end{coro}

\begin{proof}
    By Theorem \ref{thm: skewing preserves f divergence} the mapping $(p,q) \mapsto D( (1-\alpha_i) p + \alpha_i q|| (1-\bar{\alpha}) p + \bar{\alpha} q)$ is an $f$-divergence, and the result follows since the class of $f$-divergences is stable under non-negative linear combinations.
\end{proof}

We will only further pursue the case that $r = 1$, and write $S_{f,t}(\mu || \nu) \coloneqq S_{f,1,t}(\mu|| \nu)$.

We now skew,  Pearson's $\upchi^2$ divergence which we recall below.
\begin{defn}\cite{pearson1900x}
    For measures $\mu$ and $\nu$ absolutely continuous with respect to a common reference measure $d\gamma$ so that $d\mu = u d\gamma$ and $d \nu = v d\gamma$,  define
    \[
        \upchi^2( \mu; \nu) = \int \left(1- \frac{d\mu}{d \nu} \right)^2 d \nu = \int \frac{(u-v)^2}{v} d\gamma,
    \]
    and $\upchi^2(\mu; \nu) = \infty$ when $\frac{d \mu}{d \nu}$ does not exist.
\end{defn}

\begin{defn}
 For $t \in [0,1]$ and measures $\mu$ and $\nu$, define the skew $\upchi^2_t$  via:
    \[
        \upchi_t^2( \mu; \nu) = \int \frac{ \left(1- \frac{d \mu}{d \nu} \right)^2}{ t \frac{d \mu}{d \nu} + (1-t)} d \nu.
    \]
\end{defn}

Formally, the $\upchi^2$ divergence of Neyman \cite{neyman1949contribution} differs only by a notational convention $\upchi_N^2( \nu; \mu) = \upchi^2(\mu;\nu)$, see \cite{liese2006divergences} for more modern treatment and \cite{cressie1984multinomial} for background on the distances significance in statistics.  Now let us present a skew $\upchi^2$ divergence, which interpolates the Pearson and Neyman $\upchi^2$ divergences.
\begin{prop} \label{prop: chi skew properties}
    The skew $\upchi^2_t$ divergence satisfies the following,
    \begin{enumerate}
         \item \label{item: skew chi square divergence with respect to reference measure}
        When $d \mu = u d\gamma$ and $d \nu = v d\gamma$ with respect to some reference measure $\gamma$, then
        \[
            \upchi_t^2 ( \mu; \nu) = \int \frac{(u-v)^2}{tu + (1-t)v} d\gamma.
        \]
        \item \label{item: skew chi to regular chi}
        For $t=0$, $(1-t)^2 \upchi^2_t(\mu; \nu) = \upchi^2(\mu; t \mu + (1-t) \nu)$.
        \item \label{item: symmetry of chi square}
        $\upchi^2_t( \mu; \nu) = \upchi^2_{1-t}( \nu; \mu)$.
        \item \label{item: skew chi square f divergence}
        $\upchi^2_t$ is an $f$-divergence with $f(x) = (x-1)^2/(1 + t(x-1))$.
        \item \label{item: skew chi interpolates Neyman and Pearson}
        The skew $\chi_t^2$ interpolates the divergences of Neyman and Pearson, $\upchi_0^2(\mu; \nu) = \upchi^2(\mu; \nu)$ and $\upchi_1^2(\mu; \nu) = \upchi_N^2(\mu; \nu)$.
    \end{enumerate}
\end{prop}

\begin{proof}
For \eqref{item: skew chi square divergence with respect to reference measure}, the formula follows in the case $\mu \ll \nu$, from the fact that on the support of $\nu$,
\[
    \frac{u}{v} = \frac{ d \mu} {d \nu},
\]
so that
\begin{align*}
    \upchi_t^2( \mu; \nu)
        &=
            \int \frac{ \left( 1 - \frac u v \right)^2}{t \frac u v + (1-t)} v d\gamma
                \\
        &=
            \int \frac{(u-v)^2} {t u + (1-t) v} d\gamma.
\end{align*}
To prove \eqref{item: skew chi to regular chi}, we use \eqref{item: skew chi square divergence with respect to reference measure}.  Note that $d \mu = u d\gamma$ and $d \nu = v d\gamma$ implies that $d(t \mu + (1-t) \nu) = (t u + (1-t)v) d\gamma$ so that
\begin{align*}
    \upchi^2( \mu; t \mu + (1-t) \nu) 
        &=
            \int \frac{(u - (t u + (1-t) v))^2}{ t u + (1-t) v } d\gamma
                \\
        &=
            (1-t)^2 \int \frac{(u -  v)^2}{ t u + (1-t) v } d\gamma
                \\
        &=
            (1-t)^2 \upchi_t^2(\mu; \nu).
\end{align*}
It is immediate from \eqref{item: skew chi square divergence with respect to reference measure} that \eqref{item: symmetry of chi square} holds.  That $\chi_t^2$ is an $f$-divergence follows from \eqref{item: skew chi to regular chi} and Theorem \ref{thm: skewing preserves f divergence}, so that \eqref{item: skew chi square f divergence} follows.  To prove \eqref{item: skew chi interpolates Neyman and Pearson}, note that $\chi_0^2(\mu;\nu) = \chi^2(\mu;\nu)$ is immediate from the definition.  Applying this and symmetry from \eqref{item: symmetry of chi square} we have $\chi_1^2(\mu;\nu) = \chi_0^2(\nu;\mu) = \chi^2(\nu;\mu) = \chi^2_N(\mu;\nu)$.
\end{proof}

The skew divergence and skew $\upchi^2$ inherit bounds from $t=0$ case, and enjoy an interrelation unique to the skew setting as described below.

\begin{thm} \label{thm: skew chi squared related to skew relative entropy}
    For probability measures $\mu$ and $\nu$ and $t \in (0,1)$
    \begin{align} \label{eq: skew divergence versus skew chi square}
        S_t( \mu || \nu) \leq (1-t)^2 \upchi^2_t( \mu; \nu)
    \end{align}
    and
    \begin{align} \label{eq: derivative of skew divergence is skew chi}
        \frac d {dt} S_t( \mu || \nu) = (t-1) \upchi^2_t ( \mu ; \nu).
    \end{align}
\end{thm}

\begin{proof}
Recall that when $t=0$, the concavity of logarithm bounds $\log x$ by its tangent line $x-1$ so that,
\begin{align}
    \int \log\left( \frac{ d\mu}{d\nu} \right) d \mu 
        &\leq 
            \int \left( \frac{d\mu}{d \nu} - 1 \right) d \mu
                \\
        &=
            \int \left( \frac{d \mu}{d \nu} - 1 \right)^2 d \nu,
\end{align}
giving the classical bound,
\begin{align} \label{eq: classical divergence versus chi}
    D(\mu || \nu ) \leq \upchi^2( \mu; \nu).
\end{align}
Applying \eqref{eq: classical divergence versus chi} to the identities Proposition \ref{prop: divergence skew properties}, \eqref{item: skew divergence to KL divergence} and Proposition \ref{prop: chi skew properties},\eqref{item: skew chi to regular chi} gives
\begin{align}
    S_t(\mu || \nu) 
        &=
            D(\mu || t \mu + (1-t) \nu)
                \\
        &\leq
            \upchi^2( \mu ; t \mu + (1-t) \nu)
                \\
        &=
            (1-t)^2 \upchi^2_t(\mu; \nu).
\end{align}
Applying the identity 
$
(1-t)(y -1) = y - (t y+(1-t))
$
we have
\begin{align} \label{skew chi identity (to compare with total variation)}
    (1-t) \upchi_t^2 ( \mu ; \nu) 
        &= 
            \int \frac{ (\frac{d \mu}{d \nu} - 1) (\frac{d \mu}{d \nu} - (t \frac{d \mu}{d \nu}+(1-t)))}{ t \frac{d \mu}{d \nu} + (1-t)} d \nu
                \\
        &=
            \int \frac{ \frac{d \mu}{d \nu} - 1}{ t \frac{d \mu}{d \nu} + (1-t)} d \mu - \int ( \frac{d \mu}{d\nu} - 1) d \nu
                \\
        &=
            \int \frac{ \frac{d \mu}{d \nu} - 1}{ t \frac{d \mu}{d \nu} + (1-t)} d \mu.
\end{align}
Observing the expression
\[
    S_t( \mu || \nu) = \int \log \frac{d\mu}{d \nu} - \log \left(t \frac{d \mu}{d \nu} + (1-t)\right) d \mu,
\]
we compute directly,
\begin{align}
    \frac{d}{dt} S_t( \mu || \nu ) 
        &=
            - \int \frac{ \frac{d\mu}{d \nu} -1}{ t \frac{d \mu}{d \nu} + (1-t)} d \mu.
\end{align}
\end{proof}

Recall the total variation norm for a signed measure $\gamma$ to be $\sup_A \|\gamma(A)\|_{TV}$, and adopting the notation $x_+ = \max \{ x , 0 \}$ then
\[
    \| \mu - \nu \|_{TV} = \int \left( \frac{d\mu}{d \nu} - 1\right)_+ d \nu.
\]

\begin{thm} \label{thm: skews bounded by total variation}
For $\mu$ and $\nu$, and $t \in (0,1)$,
\begin{align} \label{eq: skew chi by total variation}
    \upchi^2_t( \mu ; \nu) \leq \frac{ \| \mu - \nu \|_{TV}}{t (1-t)}
\end{align}
\begin{align} \label{eq: skew divergence by total variation}
    S_t( \mu || \nu ) \leq -\log t \| \mu - \nu \|_{TV}.
\end{align}
\end{thm}

\begin{proof}
From the identity in \eqref{skew chi identity (to compare with total variation)} we have
\begin{align*}
    \upchi_t^2(\mu; \nu) 
        &=
            \frac{1}{1-t} \int \frac{ \frac{d\mu}{d \nu}( \frac{d\mu}{d \nu} - 1) }{ t  \frac{d\mu}{d \nu} + (1-t)} d \nu
                \\
        &\leq
            \frac{1}{t(1-t)} \int \left( \frac{d\mu}{d \nu} - 1 \right)_+ d \nu
                \\
        &=
            \frac{ \| \mu - \nu \|_{TV}}{t (1-t)}.
\end{align*}
    Define the function 
    \[
        \varphi(\lambda) \coloneqq S_{e^{-\lambda}}( \mu || \nu ),
    \]
    for $\lambda \in [0, \infty)$ and note that $\varphi(0) = D(\mu||\mu) = 0$.  Thus we can write
    \begin{align*}
        \varphi(\lambda) 
            &= 
                \int_0^\lambda \frac{d}{ds} S_{e^{-s}}(\mu || \nu) d s
                    \\
            &=
                \int_0^\lambda e^{-s} (1-e^{-s}) \upchi^2_{e^{-s}} ( \mu ; \nu) ds.
    \end{align*}
    Applying \eqref{eq: skew chi by total variation} gives
    \[
        S_{e^{-\lambda}}(\mu || \nu) = \varphi( \lambda) \leq \int_0^\lambda  \| \mu - \nu \|_{TV} ds = \lambda \|\mu - \nu \|_{TV}.
    \]
    The substitution $t = e^{-\lambda}$ gives \eqref{eq: skew divergence by total variation}.
\end{proof}

Observe that \eqref{eq: skew divergence by total variation} of Theorem \ref{thm: skews bounded by total variation} recovers a reverse Pinsker inequality due to Verdu \cite{verdu2014total}.

\begin{coro}[\cite{verdu2014total} Theorem 7] \label{cor: Verdu reverse Pinsker}
    For probability measures $\mu$ and $\gamma$ such that $\frac{d\mu}{d \gamma} \leq \frac 1 \beta$ with $\beta \in (0,1)$
    \[
        \| \mu - \gamma \|_{TV} \geq \frac{1 - \beta}{\log \frac 1 \beta} D( \mu || \gamma ).
    \] 
\end{coro}

\begin{proof}
    The hypothesis implies that  $ \nu = \frac{\gamma - \beta \mu}{1 - \beta}$ is a probability measure satisfying $\gamma = \beta \mu + (1-\beta) \nu$.  Applying \eqref{eq: skew divergence by total variation}
    \[
        D( \mu || \nu ) = S_\beta ( \mu || \nu) \leq - \log \beta \| \mu - \nu \|_{TV} = \frac{- \log \beta }{ 1- \beta} \| \mu - \gamma \|_{TV}.
    \]
\end{proof}
It is easily seen that the two results, \eqref{eq: skew divergence by total variation} and Theorem 7 of \cite{verdu2014total} are actually equivalent.  In contrast the proof of \eqref{eq: skew divergence by total variation} hinges on foundational properties of the divergence metrics, while Verdu leverages the monotonicity of $x \ln x/(x-1)$ for $x>1$.

\begin{proof}[Proof of Theorem \ref{thm: discrete continuous entropy inequality sharpening}]
From Proposition \ref{prop: compensation identity} 
\begin{align}
h(\sum_i p_i f_i) 
    &= 
        \sum_i p_i h(f_i) + \sum_i p_i D(f_i ||f)
            \\
    &=
        \sum_i p_i h(f_i) + \sum_i p_i S_{p_i}(f_i|| \tilde{f}_i).
\end{align}
    By Theorem \ref{thm: skews bounded by total variation}, $S_{p_i}(f_i||\tilde{f}_i) \leq \log \frac 1 {p_i} \|f_i - \tilde{f_i}\|_{TV}$. Applying H\"older's inequality completes the proof,
    \begin{align}
        \sum_i p_i S_{p_i} (f_i || \tilde{f}_i)
            &\leq
                \sum_i p_i \log \frac 1 {p_i} \| f_i - \tilde{f}_i \|_{TV}
                    \\
            &\leq
                \mathcal{T} \sum_i p_i \log \frac 1 {p_i},
    \end{align}
    where we recall $\mathcal{T} \coloneqq \sup_i \|f_i - \tilde{f}_i\|_{TV}$.
\end{proof}

Since the total variation of any two measures is bounded above by $1$ this is indeed a sharpening of \eqref{eq: discrete-continuous entropy inequality}.  Expressed in random variables it is
\begin{align}
    h_\gamma(Z) \leq \mathcal{T} H(X) + h_\gamma(Z|X).
\end{align}
which when $\gamma$ is a Haar measure and we apply to $\tilde{Z} = X+Z$ gives
\begin{align} \label{eq: seeing it as a sharpening}
    h_\gamma(X+Z) \leq \mathcal{T} H(X) + h_\gamma(Z|X),
\end{align}
while the right hand side of \eqref{eq: seeing it as a sharpening} reduces further to
\begin{align}
    h_\gamma(X+Z) \leq \mathcal{T} H(X) + h_\gamma(Z)
\end{align}
in the case that $X$ and $Z$ are independent.

Note that the quantity $h_\gamma(\sum_i p_i f_i) - \sum_i p_i h_\gamma(f_i) = \sum_{i} p_i D(f_i || f)$ can be considered a generalized Jensen-Shannon divergence,
as the case that $n=2$ and $p_1 = p_2 = \frac 1 2$ this is exactly the Jensen-Shannon Divergence.
\begin{defn}
For probability measures $\mu$ and $\nu$ define the Jensen-Shannon divergence,
\begin{align}
    JSD(\mu || \nu) = \frac 1 2 \left( D(\mu || 2^{-1} (\mu + \nu)) + D(\nu || 2^{-1} (\mu + \nu) \right).
\end{align}
\end{defn}

Theorem \ref{thm: discrete continuous entropy inequality sharpening} recovers the classical bound of the Jensen-Shannon divergence by the total variation, due to Lin, see also \cite{topsoe2000some, topsoe2003jenson} for other proofs.

\begin{coro} \cite{lin1991divergence} \label{cor: Topsoe and Lin}
    For $\mu$ and $\nu$ probability measures,
    \[
        JSD (\mu || \nu) \leq \| \mu - \nu \|_{TV} \log 2.
    \]
\end{coro}

\begin{proof}
    Apply Theorem \ref{thm: discrete continuous entropy inequality sharpening} to the Jensen-Shannon divergence, and observe that $\mathcal{T} = \|\mu - \nu \|_{TV}$ in the case of two summands.
\end{proof}

%
%
%
%
%
%
%
%
%
%
%
%
%
%
%
%
%
%
%
%
%
%

\section{Upper bounds} \label{sec: Upper bounds}

Let us state our assumptions and notations for this section.
\begin{enumerate} \label{en: Upper bounds assumptions}
    \item \label{item: X assumptions} $X$ is a random variable taking values in countable space $\mathcal{X}$, such that for $i \in \mathcal{X}$, $\mathbb{P}(X = i) = p_i$.
    \item \label{item: Z assumptions} $Z$ is an $\mathbb{R}^d$ valued random variable, with conditional densities, $f_i$ satisfying,
    \begin{align}
        \mathbb{P}(Z \in A | X=i) = \int_A f_i(z) dz = \mathbb{P}( T_i(W) \in A).
        \end{align}
        for $T_i$ a $\sqrt{\tau}$ bi-Lipschitz function, and $W$ is a spherically symmetric log-concave random vector with density $\varphi$.
    \item \label{item: seperation assumption lambda M}
        There exists $\lambda, M >0$ such that for any $i,j$,
        \begin{align}
            \#\{ k: T_{kj}(B_\lambda) \cap T_{ij}(B_\lambda) \neq \emptyset \} \leq M,
        \end{align}
        with $\#$ denoting cardinality and $T_{ij} \coloneqq T_i^{-1} \circ T_j$.
\end{enumerate}

Our assumption that $W$ is log-concave and spherically symmetric is equivalent to $W$ possessing a density $\varphi$ that is spherically symmetric in the sense that $\varphi(x) = \varphi(y)$ for $|x| =|y|$ and log-concave in the sense that $\varphi((1-t)x + ty) \geq \varphi^{1-t}(x) \varphi^t(y)$ holds for $t \in [0,1]$ and $x,y \in \mathbb{R}^d$.  By the spherical symmetry of $\varphi$, there exists $\psi:[0,\infty) \to [0,\infty)$ such that $\varphi(x) = \psi(|x|)$.
Note that by Radamacher's theorem, Lipschitz continuous functions are almost everywhere differentiable, and since bi-Lipschitz functions are necessarily invertible.   Thus, using this and \ref{prop: joint distribution of X and Z}, it follows that $Z$ has a density given by the the following expression  
\begin{align}
    f(z) = \sum_i p_i f_i(z) = \sum_i p_i \varphi(T_i^{-1}(z)) det((T_i^{-1})'(z)).
\end{align} Note that $T_i$ being $\sqrt{\tau}$ bi-Lipschitz implies $T_i^{-1}$ is $\sqrt{\tau}$ bi-Lipschitz as well, thus $T_{ij}$ is $\tau$-bi-Lipschitz, thus after potentially adjusting $T_{ij}'$ on set of measure zero, we have $ \frac 1 \tau \leq \|T_{ij}'(z)\| \leq \tau$. 
Under these assumptions we will prove the following generalization of \ref{thm: Fano Sharpening simple}.

\begin{thm} \label{thm: Fano Sharpening}
    For $X$ and $Z$ satisfying the assumptions of Section \eqref{en: Upper bounds assumptions} 
    \begin{align}
        h(Z) -  h(Z|X) \geq H(X) - \tilde{C}(W),
    \end{align}
    where $\tilde{C}$ is the following function dependent heavily on the tail behavior of $|W|$,
    \begin{align}
        \tilde{C}(W) = (M-1)(1-\mathscr{T}(\lambda \tau)) +  \mathscr{T}(\lambda)(M + h(W))  +  \mathscr{T}^{\frac 1 2}(\lambda) (\sqrt{d} + K(\varphi))
    \end{align}

with
\begin{align} \label{eq: definition of K(varphi)}
    K(\varphi) \coloneqq 
        \log \left[ \tau^d M \left( \|\varphi\|_\infty + \left(\frac{3}{\varepsilon} \right) \omega_d^{-1} \right) \right] \mathbb{P}^{\frac 1 2}(|W| > \lambda) + d \left( \int_{B_\lambda^c} \varphi(w)\log^{2}\left[ 1 + \frac{\varepsilon \tau +  \tau^2 |w|}{ \lambda}\right] dw \right)^{\frac 1 2}
\end{align}
where $\omega_d$ denoting the volume of the $d$-dimensional unit ball, $B_\lambda^c$ denotes the complement of $B_\lambda \in \mathbb{R}^d$.  
\end{thm}

Note that when $M=1$, Theorem \ref{thm: Fano Sharpening} reduces to Theorem \ref{thm: Fano Sharpening simple}.  
Additionally observe that $\log^2(x)$ is a concave function for $x \geq e$.  If one writes $m \coloneqq \max \{e, 1 +  \tau \}$ then by Jensen's inequality,
\begin{align}
    \left( \int_{B_\lambda^c} \varphi(w)\log^{2}\left[ 1 + \tau + \frac{ \tau^2 |w|}{ \lambda}\right] dw \right)^{\frac 1 2}
        &\leq
            \left( \int_{\mathbb{R}^d} \varphi(w)\log^{2}\left[ m + \frac{  \tau^2 |w|}{ \lambda}\right] dw \right)^{\frac 1 2}
                \\
        &\leq
            \log \left( m + \frac{ \tau^2 \int \varphi(w) |w| dw}{\lambda} \right).
\end{align}

Thus we can further bound
\begin{align}
    K(\varphi) \leq  \log \left[ \tau^d M \left( \|\varphi\|_\infty + \left(\frac{3}{\varepsilon} \right) \omega_d^{-1} \right) \right] \mathbb{P}^{\frac 1 2}(|W| > \lambda) + d \log \left( m + \frac{ \tau^2 \int \varphi(w) |w| dw}{\lambda} \right).
\end{align}
We now derive some implications of our assumptions on $T_{ji}$, a partitioning result on $T_{ji}(B_\lambda)$ based on the axiom of choice and for the reader's convenience we prove some elementary consequences of the boundedness of the derivatives of $T_{ij}$.
\begin{prop}\label{prop: bounds on T_ij derivative}
    For $T_{ij} = T_i^{-1} \circ T_j$ 
    \begin{align} \label{eq: mean value theorem gives ball bounds}
        T_{ij}(0) + B_{\lambda/ \tau} \subseteq T_{ij}(B_\lambda) \subseteq T_{ij}(0) + B_{\lambda \tau}
    \end{align}
    That 
    \begin{align}
       \# \{k : T_{ji}(B_\lambda) \cap T_{jk}(B_\lambda) \} \leq M 
    \end{align}
    implies that any collection  $\mathcal{X} \subseteq \mathbb{N}$ has a partition $\mathcal{X}_1, \dots, \mathcal{X}_n$, with $n \leq M$ such that $x_1, x_2 \in \mathcal{X}_k$ implies $T_{jx_1}(B_\lambda) \cap T_{jx_2}(B_\lambda) = \emptyset$.
\end{prop}

\begin{proof}
    To prove \eqref{eq: mean value theorem gives ball bounds}, observe that $T_{ij}$ $\tau$-bi-Lipschitz implies, that $T_{ij}^{-1}$ exists and is $\tau$-bi-Lipschitz as well.  Observing that $\tau$-Lipschitz implies,
    \begin{align}
        |T_{ij}(x) - T_{ij}(0)| \leq \tau |x|.
    \end{align}
    If we take $|x| < \lambda$, this inequality shows $T_{ij}(B_\lambda) \subseteq T_{ij}(0) + B_{\lambda \tau}$.  For the other inclusion, observe that since $T_{ij}^{-1}$ is $\tau$-Lipschitz as well,
    \begin{align}
        |T^{-1}_{ij}(T(0) + x)|
            &=
                |T^{-1}_{ij}(T_{ij}(0) + x) - T^{-1}_{ij}(T_{ij}(0))|
                    \\
            &\leq
                \tau |x|.
    \end{align}
    Taking $|x| < \lambda$, this shows that $T_{ij}^{-1}(T_{ij}(0) + B_{\lambda/\tau}) \subseteq B_\lambda$, which the desired inclusion follows from.\\
    
    Now we prove the existence of the partitioning.  If $M =1$, the result is obvious, and we proceed by induction.  Choose $\mathcal{X}_1$ to be a maximal subset of $\mathcal{X}$ such that $\{ T_{j x} \}_{x \in \mathcal{X}_1}$ are disjoint.  For $x_0 \in \mathcal{X} - \mathcal{X}_1$, 
    \begin{align}
        \# \{k \in \mathcal{X} - \mathcal{X}_1 : T_{ji}(B_\lambda) \cap T_{jk}(B_\lambda) \neq \emptyset \} \leq M -1.
    \end{align}
    Indeed for every $k \in \mathcal{X}$, $T_{jk}(B_\lambda)$ intersects at most $M$ others, and since $\mathcal{X}_1$ is maximal and $k \notin \mathcal{X}_1$ $T_{jk}(B_\lambda)$ must intersect one of the $T_{jx}(B_\lambda)$ for $x \in \mathcal{X}_1$ which leaves $T_{jk}(B_\lambda)$ to intersect at most $M-1$ of $\{T_{ji}(B_\lambda)\}_{i \in \mathcal{X} - \mathcal{X}_1}$.  By induction the result follows.
\end{proof}

We will need the following concentration result for the information content of a log-concave vector  \cite{Ngu13:phd, Wan14:phd, FMW16}.
\begin{thm}\label{thm: varentropy bound}
    For a log-concave density function $\varphi$ on $\mathbb{R}^d$,
    \begin{align}
        \int \left( \log \frac 1 {\varphi(x)} - h(\varphi) \right)^2 \varphi(x) dx \leq d.
    \end{align}
    where $h(\varphi)$ is the entropy of the density $\varphi)$.
\end{thm}
See \cite{FLM20} for a generalization to convex measures, which can be heavy-tailed.

The following upper bounds the sum of a sequence whose values are obtained by evaluating a spherically symmetric density at well spaced points.

\begin{lem} \label{lem: sums bounded at well spaced points}
If $\phi$ is a density on $\mathbb{R}^d$, not necessarily log-concave, given by $\phi(x) = \psi(|x|)$ for $\psi: [0,\infty) \to [0,\infty)$ decreasing, $\lambda >0$, and a discrete set $\mathcal{X} \subseteq \mathbb{R}^d$ admitting a partition $\mathcal{X}_1, \dots, \mathcal{X}_M$ such that distinct $x, y \in \mathcal{X}_k$ satisfy $|x-y| \geq 2 \lambda$,
 then there exists an absolute constant $c \leq 3$ such that
\begin{align} \label{eq: spherically symmetric unimodal sampling bound}
    \sum_{x \in \mathcal{X}} \phi(x) \leq M \left( \| \phi \|_\infty + \left( \frac c \lambda \right)^d \omega_d^{-1} \right),
\end{align}
where $\omega_d = |\{x: |x| \leq 1\}|_d$, where we recall $| \cdot |_d$ as the $d$-dimensional Lebesgue volume.  In particular, if for all $x_0 \in \mathcal{X}$
\begin{align} \label{eq: seperation assumption}
    \# \{ x \in \mathcal{X} : |x- x_0| < 2 \lambda \} \leq M,
\end{align}
then \eqref{eq: spherically symmetric unimodal sampling bound} holds.
\end{lem}

Note, that when $M=1$ and $\phi$ is the uniform distribution on a $d$-dimensional ball, this reduces to a sphere packing bound,
\begin{align}
    \# \{ \hbox{disjoint $\lambda$-balls contained in }B_{R+\lambda}\} \leq 1 + \left( \frac{Rc}{\lambda} \right)^d.
\end{align}
From which it follows, due to classical bounds of Minkowski, that $c \geq \frac 1 2$.

For the proof below we use the notations $B_\lambda(x)=\{w\in \mathbb{R}^d|~|w-x|\leq \lambda\}$ and we identify $B_\lambda\equiv  B_\lambda(0).$

\begin{proof}
Let us first see that it is enough to prove the result when $M=1$.  
\begin{align}
    \sum_{x \in \mathcal{X}} \phi(x) 
        &=
            \sum_{k=1}^M \sum_{x \in \mathcal{X}_k} \phi(x)
                \\
        &\leq \label{eq: the M=1 case}
            \sum_{k=1}^M \left( \| \phi \|_\infty + \left( \frac c \lambda \right)^d \omega_d^{-1} \right)
                \\
        &=
            M \left( \| \phi \|_\infty + \left( \frac c \lambda \right)^d \omega_d^{-1} \right)
\end{align}   
where \eqref{eq: the M=1 case} follows if we have the result when $M=1$.
We proceed in the case that $M=1$ and observe that $\psi$ non-increasing enables the following Riemann sum bound,
\begin{align}
    1 
        &=
            \int \phi(x) dx
                \\
        &\geq \label{eq: riemann sum bound}
            \sum_{k=0} \psi(k \lambda) \omega_d \lambda^d \left( (k+1)^d - k^d \right),
\end{align}
where $\omega_d \lambda^d ((k+1)^d - k^d)$ is the volume of the annulus $B_{(k+1)\lambda} - B_{k \lambda}$.  Define
\begin{align}
    \Lambda_k \coloneqq \big\{ x \in \mathcal{X} : |x| \in [k \lambda, (k+1) \lambda )  \hspace{1mm}\big\},
\end{align}
then
\begin{align} \label{eq: initial bounds on sum}
    \sum_{x \in \mathcal{X}} \phi(x) 
        &=
            \sum_{k=0}^\infty \sum_{x \in \Lambda_k} \phi(x)
                \\
        &\leq
             \sum_{k=0}^\infty \# \Lambda_k \psi(k \lambda),
\end{align}
as $\psi$ is non-increasing.
Let us now bound $\# \Lambda_k $. Using the assumption that any two elements $x$ and $y$ in $\mathcal{X}$ satisfy $|x-y|\geq 2\lambda$,
\begin{align} 
    | \bigcup_{x \in \Lambda_k} \{x + B_\lambda\}| &= \#\Lambda_k |B_\lambda|
            \\
    &= \label{eq: disjoint balls of Lambda}
        \# \Lambda_k  \omega_d \lambda^d,
\end{align}
so that it suffices to bound $| \cup_{x \in \Lambda_k} \{x +B_\lambda \}|$.
Observe that we also have
$ \cup_{x \in \Lambda_k} \{ x + B_\lambda \}$ contained in an annulus,
\begin{align}
\cup_{x \in \Lambda_k} B_\lambda(x) \subseteq \{ x: |x| \in [(k-1)\lambda, (k+2) \lambda) \},
\end{align}
which combined with \eqref{eq: disjoint balls of Lambda} gives
\begin{align}
    \# \Lambda_k \omega_d \lambda^d
        &\leq
            | \{ x: |x| \in [(k-1)\lambda, (k+2) \lambda) \} |_d
                \\
        &=
            \omega_d \lambda^d \left( (k+2)^d - (k-1)^d \right),
\end{align}
%

%
%

%

so that
\begin{align} \label{eq: the first good Lambda bound}
    \# \Lambda_k 
        &\leq 
            (k+2)^d - (k-1)^d.
\end{align}
Note the following bound, for $k \geq 1$
\begin{align} \label{eq: mean value theorem bound}
    (k+2)^d - (k-1)^d \leq 3^{d} \left((k+1)^d - k^d \right).
\end{align}
Indeed, by the mean value theorem, there exists $x_0 \in [k-1, k+2]$
\begin{align}
    (k+2)^d - (k-1)^d
        &=
            3d x_0^{d-1}
                \\
        &\leq
            3 d (k+2)^{d-1}
\end{align}
and there exists $y_0 \in [k,k+1]$ such that,
\begin{align}
    (k+1)^d - k^d
        &=
            d y_0^{d-1}
                \\
        &\geq
           d k^{d-1}.
\end{align}
Thus our bound follows from the obvious fact that for $k \geq 1$
\begin{align}
    (k+2)^{d-1} \leq 3^{d-1} k^{d-1}.
\end{align}
Compiling the above,
\begin{align}
     (k+2)^d - (k-1)^d
        &\leq
            3 d (k+2)^{d-1}
                \\
        &\leq
            3 d 3^{d-1} k^{d-1}
                \\
        &\leq
            3^d \left( (k+1)^d - k^d \right).
\end{align}
Thus for $k \geq 1$, \eqref{eq: the first good Lambda bound} and \eqref{eq: mean value theorem bound} give,
\begin{align}
    \# \Lambda_k \leq 3^d \left( (k+1)^d - k^d \right).
\end{align}
Applying this inequality to \eqref{eq: initial bounds on sum} gives
\begin{align}
    \sum_{x \in \mathcal{X}} \phi(x) 
        &\leq
            \sum_{k=0}^\infty \sum_{x \in \Lambda_k} \psi (\lambda k)
                \\
        &\leq
             \|\phi\|_\infty  + \sum_{k=1} \psi( \lambda k) \# \Lambda_{k} 
                \\
        &\leq \label{eq: since Lambda_0 has one element}
             \|\phi\|_\infty + \sum_{k=1} \psi( \lambda k) 3^d \left( (k+1)^d - k^d \right)
                \\
        &\leq
            \|\phi\|_\infty + \omega_d^{-1} \left( \frac 3 \lambda \right)^d,
\end{align}
where \eqref{eq: since Lambda_0 has one element} follows from the fact that $\# \Lambda_0 \leq 1$ (any $x \in \Lambda_0$ has $0 \in \{x + B_\lambda \}$), and the last inequality follows from the Riemann sum bound \eqref{eq: riemann sum bound}.

Let us now show that any $\mathcal{X}$ satisfying \eqref{eq: seperation assumption} necessarily satisfy \eqref{eq: spherically symmetric unimodal sampling bound}, by show that $\mathcal{X}$ satisfying \eqref{eq: seperation assumption} can be partitioned into  $M$ subsets $\{\mathcal{X}_1, \dots, \mathcal{X}_{M}\}$ such $\cup_{k =1}^{M} \mathcal{X}_k = \mathcal{X}$ and  $x, y \in \mathcal{X}_k$ implies $|x - y| \geq 2 \lambda$ for $x \neq y$.  Choose $\mathcal{X}_1$ to be a subset of $\mathcal{X}$ maximal with respect to the property $x, y \in \mathcal{X}_1$ implies $|x - y| \geq 2 \lambda$.   We note that such a maximal subset necessarily exists  by Zorn's Lemma with the partial order given by set theoretic inclusion.  Now suppose that $x_0 \in \mathcal{X}' = \mathcal{X} - \mathcal{X}_1$. Then if follows that
\begin{align}
    \# \{x \in \mathcal{X}' : |x_0 - x| < 2 \lambda \} \leq M-1.
\end{align}

Indeed, as $x_0\not \in \mathcal{X}_1$ we have from the maximality of $\mathcal{X}_1$, that there exists a $x_0'\in \mathcal{X}_1$ such that $\|x_0-x_0'\|\leq 2\lambda.$   Note that the cardinality of the set  $\{x\in \mathcal{X}: \|x-x_0\|\leq 2\lambda\}$ is at most $M$ with $x_0'\not \in \mathcal{X}':=\mathcal{X}-\mathcal{X}_1$ in the set. It follows that 
$\#\{x\in \mathcal{X}': \|x-x_0\|\leq 2\lambda\}\leq M-1.$ Applying the result inductively on $\mathcal{X}'$ our claim follows.
\end{proof}

\begin{lem} \label{lem: counting points in balls}
    For $T_{i}$  $\sqrt{\tau}$-bi-Lipshitz satisfying $\# \{l : |T_{lj}(0) - T_{ij}(0)| < 2 \lambda\} \leq M$ for all $x$, then for $\varepsilon >0$,
    \begin{align} \label{eq: the bound at w with an epsilon given a lambda}
        \# \{ l : |T_{ij}(w) - T_{lj}(w)| < \varepsilon \} \leq M \left( \frac{\lambda + \varepsilon + 2 \tau |w|}{\lambda}\right)^d.
    \end{align}
    In particular,
    \begin{align} \label{eq: the bound at zero with an epsilon}
        \# \{ l : |T_{ij}(0) - T_{lj}(0)| < \varepsilon \} \leq M \left( \frac{ \varepsilon + \lambda}{\lambda} \right)^d.
    \end{align}
\end{lem}

\begin{proof}
    We will first prove and then leverage  \eqref{eq: the bound at zero with an epsilon}. Note that there is nothing to prove when $\varepsilon \leq 2\lambda$ as \eqref{eq: the bound at zero with an epsilon} is weaker than the assumption.  When $2\lambda < \varepsilon$ choose (by Zorn's lemma for instance) $\Lambda$ to be a maximal subset of $\mathbb{N}$ such that for $k \in \Lambda$, $|T_{k j}(0) - T_{ij}(0)| < \varepsilon$ and $|T_{k j}(0) - T_{k' j} (0)| \geq 2\lambda$ for $k, k' \in \Lambda$, with $k \neq k'$.  By construction, for a fixed $j$,  $T_{k j}(0) + B_\lambda$ are disjoint over $k \in \Lambda$ contained in $T_{ij}(0) + B_{\lambda + \varepsilon}$. Thus
    \begin{align}
        \lambda^d \omega_d \# \Lambda
            &= 
                |\cup_{k \in \Lambda} \{ T_{k j}(0) + B_\lambda \}| 
                    \\
            &\leq
                |\{ T_{ij}(0) + B_{\lambda + \varepsilon} \}|
                    \\
            &=
                (\lambda + \varepsilon)^d \omega_d,
    \end{align}
    and we have the following bound on the cardinality of $\Lambda$,
    \begin{align} \label{eq: cardinality of Lambda}
        \# \Lambda \leq \left(\frac{\lambda+ \varepsilon}{\lambda} \right)^d.
    \end{align}
    Applying \eqref{eq: cardinality of Lambda}, the assumed cardinality bounds,  and the maximality of $\Lambda$, which implies$\cup_{k \in \Lambda} \{T_{k j}(0) + B_\lambda\}$ contains every $T_{lj}(0)$ such that $|T_{lj}(0)-T_{ij}(0)| < \varepsilon$, we have
    \begin{align}
        \# \{ l : |T_{ij}(0) - T_{lj}(0)| < \varepsilon \}
            &\leq
                \sum_{k \in \Lambda} \# \{ m : |T_{mj}(0) - T_{k j}(0)| < \lambda \} 
                    \\
            &\leq
                 M \left(\frac{\lambda+ \varepsilon}{\lambda} \right)^d.
    \end{align}
        
    Towards \eqref{eq: the bound at w with an epsilon given a lambda},  by the mean value theorem, there exists $t \in [0,1]$ such that
    \begin{align}
        T_{ij}(w) - T_{lj}(w) = T_{ij}(0) - T_{lj}(0) + (T_{ij}'(tw) - T_{lj}'(tw))w.
    \end{align}
    Note that if $|T_{ij}(w) - T_{lj}(w)| < \varepsilon$, then
    \begin{align}
        |T_{ij}(0) - T_{lj}(0)| 
            &=
                | T_{ij}(w) - T_{lj}(w) - (T_{ij}'(tw) - T_{lj}'(tw))w|
                    \\
            &\leq
                | T_{ij}(w) - T_{lj}(w)| + |(T_{ij}'(tw) - T_{lj}'(tw))w|
                    \\
            &\leq
                \varepsilon + 2 \tau |w|.
    \end{align}
    Thus
    \begin{align}
        \# \{ l : |T_{ij}(w) - T_{lj}(w)| < \varepsilon \} \leq \# \{ l : |T_{ij}(0) - T_{lj}(0)| < \varepsilon + 2 \tau |w| \}.
    \end{align}
    Applying \eqref{eq: the bound at zero with an epsilon},
    \begin{align}
        \# \{ l : |T_{ij}(w) - T_{lj}(w)| < \varepsilon \} \leq M \left( \frac{\lambda + \varepsilon + 2 \tau |w|}{\lambda}\right)^d.
    \end{align}
\end{proof}

\begin{coro} \label{cor: bounds for the sums}
For a density $\phi(w) = \psi(|w|)$, with $\psi$ decreasing, $\varepsilon >0$, $T_i$ $\sqrt{\tau}$-bi-Lipschitz, $T_{ij} = T_i^{-1} \circ T_j$, such that there exists $M \geq 1$ such that for any $i$
\begin{align}
    |\{k: T_{ij}(B_\lambda) \cap T_{kj}(B_\lambda) \neq \emptyset \}| \leq M,
\end{align} 
then
\begin{align}
    \sum_i \phi(T_{ij}(w)) det(T_{ij}'(w)) \leq \tau^d M \left( 1 + \frac{\varepsilon \tau +  \tau^2 |w|}{ \lambda}\right)^d \left( \|\phi\|_\infty + \left(\frac{3}{\varepsilon} \right)^d \omega_d^{-1} \right).
\end{align}
\end{coro}

\begin{proof}
 For any $k$ suppose $|T_{ij}(0) - T_{kj}(0)| < 2 \lambda/ \tau$, then $\{T_{ij}(0) + B_{\lambda/\tau} \} \cap \{T_{kj}(0) + B_{\lambda/\tau} \} \neq \emptyset$ which by Proposition \ref{prop: bounds on T_ij derivative} implies $T_{ij}(B_\lambda) \cap T_{kj}(B_\lambda) \neq \emptyset$.  Thus, $\#\{k : |T_{ij}(0) - T_{kj}(0)| < 2 \lambda/ \tau \} \leq M$.  By Lemma \ref{lem: counting points in balls},
 \begin{align}
     \# \{l : |T_{ij}(w) - T_{lj}(w)| < 2\varepsilon \} 
        &\leq 
            M \left( \frac{ \frac{2 \lambda}{\tau} + 2\varepsilon + 2 \tau |w|}{ \frac{ 2 \lambda}{\tau}} \right)^d
                \\
        &=
            M \left( 1 + \frac{\varepsilon \tau +  \tau^2 |w|}{ \lambda}\right)^d. \label{eq: bound on some guys}
 \end{align}
 
This shows that \eqref{eq: seperation assumption} holds for $x_i = T_{ij}(w)$, $\lambda = \varepsilon$  and $M$ in \eqref{eq: seperation assumption} identified with $M \left( 1 + \frac{\varepsilon \tau + \tau^2 |w|}{\lambda} \right)^d$. It follows from Lemma \ref{lem: sums bounded at well spaced points} that
 \begin{align}
     \sum_i \phi(T_{ij}(w)) \leq M \left( 1 + \frac{\varepsilon \tau +  \tau^2 |w|}{ \lambda}\right)^d \left( \|\phi\|_\infty + \left(\frac{3}{\varepsilon} \right)^d \omega_d^{-1} \right).
 \end{align}
This, combined with $\|T_{ij}(x)\| \leq \tau$ giving the determinant bounds $det(T_{ij}'(w)) \leq \tau^d$ yields,
\begin{align}
    \sum_i \phi(T_{ij}(w)) det(T_{ij}'(w))
        &\leq 
            \tau^d\sum_i \phi(T_{ij}(w))
                \\
        &\leq
            \tau^d M \left( 1 + \frac{\varepsilon \tau +  \tau^2 |w|}{ \lambda}\right)^d \left( \|\phi\|_\infty + \left(\frac{3}{\varepsilon} \right)^d \omega_d^{-1} \right).
\end{align}

\end{proof}

\begin{coro} \label{cor: The bound C}
Consider $T_{i}$, $\sqrt{\tau}$-bi-Lipschitz and suppose there exists $M, \lambda>0$ such that $\# \{k : T_{kj}(B_\lambda) \cap T_{ij}(B_\lambda) \neq \emptyset  \} \leq M $ holds for any $i,j$. Then for a spherically symmetric log-concave density $\varphi(x) = \psi(|x|)$,
 \begin{align}
     \int_{B_\lambda^c} \varphi(w) \log &\left( \sum_i p_i \sum_j \varphi(T_{ij}(w)) det(T_{ij}'(w)) \right) \leq    K(\varphi) \mathbb{P}^{\frac 1 2} ( |W| > \lambda),
 \end{align}
 where 
 \begin{align} \label{eq: definition of C}
    K(\varphi) &\coloneqq 
        \log \left[ \tau^d M \left( \|\varphi\|_\infty + \left(\frac{3}{\lambda} \right) \omega_d^{-1} \right) \right] \mathbb{P}^{\frac 1 2}(|W| > \lambda) + d \left( \int_{B_\lambda} \varphi(w)\log^{2}\left[ 1 + \tau + \frac{ \tau^2 |w|}{ \lambda}\right] dw \right)^{\frac 1 2}.
\end{align}
\end{coro}
Note that $K(\varphi)$ depends on only on the statistics of $\varphi$, its maximum and its a logrithmic scaling of its norm, which can be easily further bounded by concavity results. The proof does not leverage log-concavity, a stronger assumption than $\psi$ non-increasing\footnote{Under spherical symmetry and log-concavity $\|\varphi\|_\infty = \varphi(0)$.  Indeed, $\varphi(0) = \varphi(\frac{-x +x}{2}) \geq \sqrt{\varphi(-x) \varphi(x)} = \varphi(x)$.  Using log-concavity again for $t \in (0,1)$, $\psi(t|x|) = \varphi((1-t) 0 + tx) \geq \varphi^{1-t}(0) \varphi^{t}(x) \geq \varphi(x) = \psi(|x|)$.  Thus it follows that $\psi$ is non-increasing.}, except to ensure that the relevant statistics are finite.

\begin{proof}
Note that,
\begin{align}
    \int_{B_\lambda^c} \varphi(w) \log &\left( \sum_i p_i \sum_j \varphi(T_{ij}(w)) det(T_{ij}'(w))  \right) dw
        \\
        &\leq \int \varphi(w) \mathbbm{1}_{B_\lambda^c} \log\left(\tau^d M \left( 1 +  \tau +  \frac{\tau^2 |w|}{ \lambda}\right)^d \left( \|\varphi\|_\infty + \left(\frac{3}{\lambda} \right)^d \omega_d^{-1} \right)  \right) dw
            \\
        &\leq
             K(\varphi) \hspace{1mm} \mathbb{P}^{\frac 1 2}(|W| > \lambda)
\end{align}
where the first inequality is an application of Corollary \ref{cor: bounds for the sums} with $\varepsilon = \lambda$ and the second from Cauchy-Schwartz.
\end{proof}

\begin{proof}[Proof of Theorem \ref{thm: Fano Sharpening}]
Using $f_i(x) = \varphi( T_i^{-1}(x) ) det((T_i^{-1})'(x))$ and applying the substitution $x= T_i(w)$ we can write
\begin{align}
\int &f_i(x) \log \left( 1 + \frac{ \sum_{j\neq i} p_j f_j(x)}{p_i f_i(x)} \right) dx 
    \\
        &=
            \int \varphi( T_i^{-1}(x) ) det((T_i^{-1})'(x)) \log \left( 1 + \frac{ \sum_{j\neq i} p_j \varphi( T_j^{-1}(x) ) det((T_j^{-1})'(x))}{p_i \varphi( T_i^{-1}(x) ) det((T_i^{-1})'(x))} \right) dx 
                \\
        &=
            \int \varphi(w)  \log \left( 1 + \frac{ \sum_{j\neq i} p_j \varphi(T_j^{-1}(T_i(w))) det((T_j^{-1})'(T_i(w)))det(T_i'(x))}{p_i \varphi(w) } \right) dw
                \\
        &=
            \int \varphi(w) \log \left( 1 + \frac{ \sum_{j\neq i} p_j \varphi(T_{ji}(w)) det(T_{ji}'(w))}{p_i \varphi(w)} \right) dw
\end{align}
Thus, applying Jensen's inequality 
    \begin{align}
        \sum_i p_i \int f_i(x) \log &\left( 1 + \frac{ \sum_{j\neq i} p_j f_j(x)}{p_i f_i(x)} \right) dx 
            \\
            &=
                \sum_i p_i \int \varphi(w) \log \left( 1 + \frac{ \sum_{j\neq i} p_j \varphi(T_{ji}(w)) det(T_{ji}'(w))}{p_i \varphi(w)} \right) dw
                    \\
            &\leq 
                \int \varphi(w) \log \left( 1 + \frac{ \sum_i \sum_{j\neq i} p_j \varphi(T_{ji}(w)) det(T_{ji}'(w) )}{ \varphi(w)} \right) dw
                    \\
            &=
                \int \varphi(w) \log \left( 1 + \frac{ \sum_j p_j \sum_{i\neq j} \varphi(T_{ji}(w)) det(T_{ji}'(w) )}{ \varphi(w)} \right) dw
    \end{align}
    We will split the integral in to two pieces.  Using $\log(1+x) \leq x$ on $B_\lambda$
    \begin{align}
        \int_{B_\lambda} &\varphi(w) \log \left( 1 + \frac{ \sum_j p_j \sum_{i\neq j} \varphi(T_{ji}(w)) det(T_{ji}'(w) )}{ \varphi(w)} \right) dw
            \\
            &\leq
                \int_{B_\lambda} \sum_j p_j \sum_{i\neq j} \varphi(T_{ji}(w)) det(T_{ji}'(w)) dw
                    \\
            &=
               \sum_j p_j \sum_{i\neq j} \int_{B_\lambda}  \varphi(T_{ji}(w)) det(T_{ji}'(w)) dw
                \\
            &=
                \sum_j p_j \left(\sum_{i \neq j} \int_{T_{ji}(B_\lambda)} \varphi(x) dx \right)
                    \\
            &=
                \sum_j p_j \left(\sum_{\{i \neq j: T_{ji}(B_\lambda) \cap B_\lambda \neq \emptyset\}} \int_{T_{ji}(B_\lambda)} \varphi(x) dx + \sum_{\{i: T_{ji}(B_\lambda) \cap B_\lambda = \emptyset\}} \int_{T_{ji}(B_\lambda)} \varphi(x) dx \right)
                          \\
            &\leq  \label{eq: simple inclusion}
                \sum_j p_j \left(\left(\sum_{\{i \neq j: T_{ji}(B_\lambda) \cap B_\lambda \neq \emptyset\}} \int_{T_{ji}(0)+B_{\lambda \tau}} \varphi(x) dx \right) + M \int_{B_\lambda^c} \varphi(x) dx \right)
                    \\
             &\leq \label{eq: integral at center is maximal}
                \sum_j p_j \left( (M-1) \int_{B_{\lambda \tau}} \varphi(x) dx + M \int_{B_{\lambda}^c} \varphi(x) dx \right) 
                    \\
            &=
                (M-1) \mathbb{P}(|W| \leq \lambda \tau) + M \mathbb{P}(|W| > \lambda)
    \end{align}
    where inequality \eqref{eq: simple inclusion} follows from Proposition \ref{prop: bounds on T_ij derivative} and inequality \eqref{eq: integral at center is maximal} follows another application of Proposition \ref{prop: bounds on T_ij derivative} and the fact that the map $s(x) = \int_{x + B_{\lambda \tau}} \varphi(z)dz$ is maximized at $0$.  To see this, observe that $s$ can be realized as the convolution of two spherically symmetric unimodal functions, explicitly $s(x) = \varphi * \mathbbm{1}_{B_{\lambda \tau}}(x)$.  Since the class of such functions is stable under convolution, see for instance \cite[Proposition 8]{li2019further}, $s$ is unimodal and spherically symmetric which obviously implies $s(0) \geq s(x)$ for all $x$. 
    
    Using the fact that $T_{ii}(w) = w$,
    \begin{align}
     \int_{B_\lambda^c} &\varphi(w) \log \left( 1 + \frac{ \sum_j p_j \sum_{i\neq j} \varphi(T_{ji}(w)) det(T_{ji}'(w))}{ \varphi(w)} \right) dw
        \\
        &=
        \int_{B_\lambda^c} \varphi(w) \log \left(   \sum_j p_j \sum_{i} \varphi(T_{ji}(w)) det(T_{ji}'(w)) \right) dw - \int_{B_\lambda^c} \varphi(w) \log \varphi(w) dw
            \\
        &\leq
           K(\varphi) \mathbb{P}^{\frac 1 2}(|W| \leq \lambda) - \int_{B_\lambda^c} \varphi(w) \log \varphi(w) dw,
    \end{align}
    where the bound $K(\varphi)$ is defined from Corollary \ref{cor: The bound C}.  By Cauchy-Schwartz, followed by Theorem \ref{thm: varentropy bound},
    \begin{align}
        -\int_{B_\lambda^c} &\varphi(w) \log \varphi(w) dw
            \\
            &=
        h(W) \mathbb{P}(|W| \geq \lambda) + \int_{B_\lambda^c} \varphi(w) \left( \log \frac{ 1} {\varphi(w)} - h(W) \right)dw 
            \\
            &\leq
                h(W) \mathbb{P}(|W| \geq \lambda) + \sqrt{\int(\log \frac 1 {\varphi(w)} - h(\varphi))^2 \varphi(w) dw \int_{B_\lambda^c} \varphi(w) dw } 
                    \\
            &\leq
                 h(W) \mathbb{P}(|W| \geq \lambda) + \sqrt{d} \hspace{1 mm} \mathbb{P}^{\frac 1 2}(|W| \geq \lambda).
    \end{align}
\end{proof}

\subsection{Commentary on Theorem \ref{thm: Fano Sharpening}}

Let us comment on the nature of $\mathbb{P}(|W| \geq \lambda)$ for $W$ log-concave.
 It is well known that in broad generality (see \cite{Bor73a, LS93, Gue99, Sergey2016ProbabilitySurveys,Sergey2015localization}), log-concave random variables satisfy ``sub-exponential'' large deviation inequalities.  The following is enough to suit our needs.
 
 \begin{lem}  \cite[Theorem 2.8]{LS93} \label{lem:large deviations for log-concave measures}
     When $W$ is a log-concave and $t,r >0$, then 
     \[
        \mathbb{P}( |W| > rt) \leq \mathbb{P}( |W| > t )^{\frac{r+1}{2}}
     \]
 \end{lem}
 
 \begin{coro} \label{cor:deviation bounds for log-concave random variables}
     For a spherically symmetric log-concave random vector $W$ such that  $\mathbb{E}W_1^2 = \sigma^2$, where $W_1$ is the random varaiable given by the first coordinate of $W$,
        \[
            \mathbb{P}(|W| >t) \leq C e^{-ct},
        \]
        where $C = 2^{-1/2}$, $c = \frac{\log 2}{2 \sqrt{2 d \sigma^2}}$.
 \end{coro}
 
 \begin{proof}
     By Chebyshev's inequality $\mathbb{P}(|W| > \sqrt{2 d \sigma^2}) \leq \frac 1 2$.  Hence for $r > 1$, by Lemma \ref{lem:large deviations for log-concave measures} 
     \begin{align*}
         \mathbb{P}( |W| > r \sqrt{2 d \sigma^2})
            &\leq 
                \mathbb{P}( |W| >  \sqrt{2 d \sigma^2})^{\frac{r+1}{2}}
                    \\
            & \leq 
                2^{- \frac{r+1}{2}}.
     \end{align*}
     Taking $t = r \sqrt{2d \sigma^2}$ gives the result.
 \end{proof}
 

 \begin{lem}\label{lem: Caffarelli}
  Suppose that the spherically symmetric log-concave $W$ is strongly log-concave in the sense that its density function $\varphi$ satisfies $\varphi((1-t)x+ty) \geq e^{ t(1-t)|x-y|^2/2} \varphi^{1-t}(x) \varphi^t(y)$, for $x,y \in \mathbb{R}^d$ and $t \in [0,1]$, then
  \begin{align}
      \mathbb{P}(|W|> t) \leq \mathbb{P}(|\mathcal{Z}| > t),
  \end{align}
  where $\mathcal{Z}$ is a standard normal vector.
 \end{lem}
 
 \begin{proof}
    By the celebrated result of Caffarelli \cite{Caf00}, $W$ can be expressed as $T(\mathcal{Z})$ where $T$ is the necessarily $1$-Lipschitz Brenier transportation map.  Moreover it follows from the assumed spherically symmetry of $W$, the radial symmetry of the Gaussian, and \cite[Lemma 1]{Caf00} that $T$ is spherically symmetric.  In particular $T(0) = 0$, so that $|\mathcal{Z}| \leq t$ implies
$|T(\mathcal{Z})| = |T(\mathcal{Z})- T(0)| \leq |\mathcal{Z}-0| \leq t$, and our result follows.
\end{proof}

 \begin{coro} \label{cor: tail bounds for log-concave and strongly}
   Suppose $X$ and $Y$ are variables satisfying the conditions of Section \ref{sec: Upper bounds} for $\tau, M =1$, $\lambda$, and $W$ possessing spherically symmetric log-concave density $\varphi$. Then
    \begin{align}
        H(X|Y)
            &\leq \tilde{C} \hspace{1mm} 2^{- \frac{1 + \lambda/\sqrt{2 \sigma^2 d}}{4}}.
    \end{align}
    Furthermore, if $W$ is strongly log-concave then
    \begin{align}
        H(X|Y)
            &\leq \tilde{C} \hspace{1mm} \mathbb{P}^{\frac 1 2} (|\mathcal{Z}| > t),
    \end{align}
    where $\mathcal{Z}$ is a standard Gaussian vector and $\tilde{C} =\left(1+ \sqrt{d} +  h(W) + K(\varphi) \right)$, with $K(\varphi)$ defined as in \eqref{eq: definition of K(varphi)}.
 \end{coro}

\begin{proof}
    The proof is an immediate application of Corollary \ref{cor:deviation bounds for log-concave random variables} and Lemma \ref{lem: Caffarelli} to Theorem \ref{thm: Fano Sharpening}.
\end{proof}

\section{Applications} \label{sec: applications}
Mixture distributions are ubiquitous, and their entropy is a fundamental quantity.  
We now  give special attention to how the ideas of Section~\ref{sec: Upper bounds} can be sharpened in the case that $W$ is a Gaussian.

\begin{prop}\label{prop: Gaussian example for upper bounds}
    When $X$ and $Y$ satisfy the assumptions of Section \ref{sec: Upper bounds}, for $\tau$, $M$, $\lambda$, $T_{ij}$ and $W \sim \varphi(w) = e^{-|x|^2}/ (2 \pi)^{d/2}$, then 
    \begin{align}
        H(X|Y) \leq (M-1) \mathbb{P}( |W| \leq \tau \lambda) + J_d(\varphi) \mathbb{P}(|W| > \lambda)
    \end{align}
     with 
    \begin{align}
          J_d(\varphi) =    \log \left[ e^{(\lambda/\sigma)^2 + M} (\tau e )^d M  \left( 1 + \tau + \tau^2 + \frac{\tau^2  d \sigma}{ \lambda} \right)^d  \left( 1 + \left( \frac{3\sqrt{2 \pi}\sigma}{\lambda } \right)^d \omega_d^{-1} \right) \right]
    \end{align} 
    for $d \geq 2$ and
    \begin{align}
     J_1(\varphi) =     \log \left[ e^{(\lambda/\sigma)^2 + M+2} \tau M    \left( 1 + \tau + \tau^2 +  \frac{\tau^2 \sigma }{ \lambda^2} \right)  \left( 1 +  \sqrt{\frac{9 \pi}{2}} \frac{\sigma}{ \lambda }  \right) \right],
     \end{align}
     when $d=1$.
\end{prop}
The proof of the case $d \geq 2$ is given below, a similar argument in the $d =1$ case is given in the appendix as Proposition \ref{prop: Appendix d=1 bounds}.
\begin{proof}
As in the proof of Theorem \ref{thm: Fano Sharpening},
\begin{align}
    H(X|Y) 
        =& \sum_i p_i \int \varphi(w) \log \left( 1 + \frac{ \sum_{j \neq i} p_j \varphi(T_{ji}(w)) det(T_{ji}'(w)) }{p_i \varphi(w)} \right) dw
            \\
        \leq&
           \int_{B_\lambda^c} \varphi(w) \log \left( 1 + \frac{ \sum_{j} p_j  \sum_{i \neq j} \varphi(T_{ji}(w)) det(T_{ji}'(w)) }{\varphi(w)} \right) dw 
                \\
        &+  \int_{B_\lambda} \varphi(w) \log \left( 1 + \frac{ \sum_{j} p_j  \sum_{i \neq j} \varphi(T_{ji}(w)) det(T_{ji}'(w)) }{\varphi(w)} \right) dw.
\end{align}
We use the general bound from the proof of Theorem \ref{thm: Fano Sharpening} for
\begin{align}
     \int_{B_\lambda} \varphi(w) \log &\left( 1 + \frac{ \sum_{j} p_j  \sum_{i \neq j} \varphi(T_{ji}(w)) det(T_{ji}'(w)) }{\varphi(w)} \right) dw 
        \\
        &\leq 
        (M-1) \mathbb{P}(|W| \leq \lambda \tau) + M \mathbb{P}(|W| > \lambda). \label{eq: Gaussian outside the ball piece}
\end{align}
Splitting the integral,
\begin{align}
\int_{B_\lambda^c} &\varphi(w) \log \left( 1 + \frac{ \sum_j p_j \sum_{i\neq j} \varphi(T_{ji}(w)) det(T_{ji}'(w))}{ \varphi(w)} \right) dw
        \\
        &=
            \int_{B_\lambda^c} \varphi(w) \log \left(   \sum_j p_j \sum_{i} \varphi(T_{ji}(w)) det(T_{ji}'(w)) \right) dw - \int_{B_\lambda^c} \varphi(w) \log \varphi(w) dw.
\end{align}
Using Corollary \ref{cor: bounds for the sums},  Jensen's inequality, and then Proposition \ref{prop: appendix regular norm bound},
\begin{align}
    \int_{B_\lambda^c} \varphi(w) \log &\left(   \sum_j p_j \sum_{i} \varphi(T_{ji}(w)) det(T_{ji}'(w)) \right) dw 
    \\
    &\leq
            \int_{B_\lambda^c} \varphi(w) \log \left[ \tau^d M \left( 1 + \tau + \frac{\tau^2 |w|}{\lambda} \right)^d \left( \|\varphi\|_\infty + \left( \frac 3 {\lambda} \right)^d \omega_d^{-1} \right) \right] 
                \\
    &\leq 
        \mathbb{P}(|W| > \lambda) \log \left[ \tau^d M \left( 1 + \tau + \frac{\tau^2 \int \mathbbm{1}_{\{|w| > \lambda \}}\varphi(w) |w|dw}{\mathbb{P}(|W|> \lambda) \lambda} \right)^d  \left( \|\varphi\|_\infty + \left( \frac{3}{\lambda} \right)^d \omega_d^{-1} \right) \right]
            \\
    &\leq
        \mathbb{P}(|W| > \lambda) \log \left[ \tau^d M \left( 1 + \tau + \frac{\tau^2 (\lambda + \sigma d)}{ \lambda} \right)^d  \left( \|\varphi\|_\infty + \left( \frac{3}{\lambda} \right)^d \omega_d^{-1} \right) \right]. \label{eq: the sum piece of the complement bound}
\end{align}

Then applying Proposition \ref{prop: entropy tail bounds}
\begin{align}
   - \int_{B_\lambda^c}  \varphi(w) \log \varphi(w) dw \leq \left( \frac d 2 \log 2 \pi e^2 \sigma^2 + \lambda^2/\sigma^2 \right) \hspace{.5 mm} \mathbb{P}(|W| > \lambda) \label{eq: entropy tail bound gaussian}
\end{align}

Combining (\ref{eq: Gaussian outside the ball piece}), (\ref{eq: the sum piece of the complement bound}), and  (\ref{eq: entropy tail bound gaussian}) we have
\begin{align}
   H(X|Y) \leq (M-1) \mathbb{P}(|W| \leq \lambda \tau)  + J_d(\varphi) \mathbb{P}(|W| > \lambda)
\end{align}
with 
\begin{align}
    J_d(\varphi) &=    \log \left[ e^{(\lambda/\sigma)^2 + M} (\tau e \sigma \sqrt{2 \pi})^d M  \left( 1 + \tau + \tau^2 + \frac{ \tau^2 \sigma d}{ \lambda} \right)^d  \left( (2 \pi \sigma^2)^{- \frac d 2} + \left( \frac{3}{\lambda } \right)^d \omega_d^{-1} \right) \right]
        \\
        &=
        \log \left[ e^{(\lambda/\sigma)^2 + M} (\tau e  \sqrt{2 \pi})^d M  \left( 1 + \tau + \tau^2 + \frac{ \tau^2 \sigma d}{ \lambda} \right)^d  \left( (2 \pi )^{- \frac d 2} + \left( \frac{3 \sigma}{\lambda } \right)^d \omega_d^{-1} \right) \right].
\end{align}
\end{proof}

For example, when $X$ takes values $\{x_i\} \in \mathbb{R}$ such that $|x_i - x_j| \geq 2 \lambda$ and $Y$ is given by $X+W$ where $W$ is independent Gaussian noise with variance $\sigma^2$, then $Y$ has density $\sum_i p_i f_i(y)$ with $f_i(y) = \varphi_\sigma(y-x_i)$. Let $T_i(y) = y-x_i$.  Thus $T_{ij}(B_\lambda) = B_\lambda + x_j-x_i$, so that $\{ T_{kj}(B_\lambda) \}_k$ are disjoint and we can take $M=1$ and $\tau =1$. Applying Proposition \ref{prop: Gaussian example for upper bounds}, we have
\begin{align}
    H(X|Y) = H(X| X+W)\leq J_1(\varphi) \mathbb{P}(|W|> \lambda) = J_1(\varphi) \mathbb{P}(|\mathcal{Z}| > \lambda/\sigma)
\end{align}
with
\begin{align}
     J_1(\varphi) =     \log \left[ e^{(\lambda/\sigma)^2 + 3}    \left(3 +  \frac{ \sigma }{ \lambda^2} \right)  \left( 1 +   \sqrt{\frac{9 \pi }{2}} \frac{\sigma}{ \lambda }  \right) \right].
\end{align}
Notice that for $t>0$, $tX$ and $tX+tW$ satisfy the same conditions with $\tilde{\lambda} =  t \lambda$ and $\tilde{\sigma} = t \sigma$, and since $H(tX | tX + tW) = H(X|X+W)$, after applying the result to $tX, tX + tW$ we have
\begin{align}
    H(X|X+ W) 
        &=
           H(tX |tX + tW)
                \\
        &\leq
            \log \left[ e^{(\lambda/\sigma)^2 + 3}    \left(3 +  \frac{ \sigma }{ t \lambda^2} \right)  \left( 1 +   \sqrt{\frac{9\pi}{2}} \frac{\sigma}{ \lambda }  \right) \right] \mathbb{P}(|\mathcal{Z}| > \lambda/\sigma).
\end{align}
Taking the limit with $t \to \infty$, we obtain,
\begin{align}
    H(X|X+W)
        \leq 
           \log \left[3 e^{(\lambda/\sigma)^2 + 3}     \left( 1 +   \sqrt{\frac{9\pi}{2}} \frac{\sigma}{ \lambda }  \right) \right] \mathbb{P}(|\mathcal{Z}| > \lambda/\sigma).
\end{align}
We collect these observations in the following Corollary.

\begin{coro}\label{cor: AGWN in 1 d bounds}
    When $X$ is a discrete $\mathbb{R}$ valued random variable taking values $\{x_i\}$ such that $|x_i - x_j| \geq 2 \lambda$ for $i \neq j$ and $W$ is an independent Gaussian variable with variance $\sigma^2$, then
        \begin{align}
            H(X|X+W)  \leq \log \left[3 e^{(\lambda/\sigma)^2 + 3}     \left( 1 +   \sqrt{\frac{9\pi}{2}} \frac{\sigma}{ \lambda }  \right) \right] \mathbb{P} \left(|\mathcal{Z}| > \lambda/\sigma \right).
        \end{align}
\end{coro}

\subsection{Fano's inequality}
A multiple hypothesis testing problem is described in the following, with $X$ an index $i \in \mathcal{X}$ is drawn and then samples are drawn from the distribution $f_i$, with a goal of determining the value $i$.   If $Z$ denotes  a random variable with  $P(Z \in A |X = i) = \int_A f_i(z) dz$, then by the commutativity of mutual information proven in Proposition \ref{prop: mutual information equality}, $H(X | Z) =  h(Z|X) - h(Z) + H(X)$.  Thus bounds on the mixture distribution  are equivalent to bounds $H(X|Z)$.  For $\hat X = g(Z)$, Fano's inequality provides  the following bound
\begin{align} \label{eq: Fano's inequality}
H(X|Z) \leq H(e) + \mathbb{P}(e)\log (\#\mathcal{X}-1)
\end{align} 
where $e = \{\hat{X} \neq X \}$ is the occurrence of an error.  Fano and Fano-like inequalities are important in multiple hypothesis testing, as they can be leveraged to deliver bounds on the  Bayes risk (and hence min/max risk); we direct the reader to \cite{Gun11,birge2005new, YB99} for more background.  Fano's inequality gives a lower bound on the entropy of a mixture distribution, that can also give a non-trivial improvement on the concavity of entropy through the equality $H(X|Z) = H(X) + h(Z|X) - h(Z)$.  Combined with \eqref{eq: Fano's inequality},
\begin{align} \label{eq: theorem 1.1 stated for our corollary}
    h(\sum_i p_i f_i) - \sum_i p_i h(f_i) \geq H(p) - \left( H(e) + \mathbb{P}(e) \log ( |\mathcal{X}| - 1) \right).
\end{align}

In concert with with Theorem \ref{thm: discrete continuous entropy inequality sharpening} we have the following corollary.
\begin{coro}
 For $X$ distributed on indices $i \in \mathcal{X}$, and $Z$ such that $Z|\{X=i\}$ is distributed according to $f_i$, then given an estimator $\tilde X = f(Z)$, with $e = \{ X \neq \tilde{X} \}$
    \begin{equation}
       (1 - \mathcal{T}_f) H(X) \leq H(e) + \mathbb{P}(e) \log( |\mathcal{X}| - 1).
    \end{equation}
\end{coro}

\begin{proof}
By Fano's inequality $H(e) + \mathbb{P}(e) \log(N-1) \geq H(X|Z)$.  Recalling that $H(X|Z) = H(X) - (h(\sum_i p_i f_i) - \sum_i p_i h(f_i))$ and by Theorem \ref{thm: discrete continuous entropy inequality sharpening}
\begin{align} 
    H(X) - (h(\sum_i p_i f_i) - \sum_i p_i h(f_i))
        &\geq 
            H(X) - \mathcal{T}_f H(X),
\end{align}
gives our result.
\end{proof}
Heuristically, this demonstrates that ``good estimators'' are only possible for hypothesis distributions discernible in total variation distance. 
For example in the simplest case of binary hypothesis testing where  $n=2$, the inequality is
$ (1- \|f_1 - f_2\|_{TV}) H(X) \leq H(e)$, demonstrating that existence of an estimator improving on the trivial lower bound $H(X)$, is limited explicitly by the total variation distance of the two densities.\\

We note that the pursuit of good estimators $\hat{X}$ is a non-trivial problem in most interesting cases, so much so that Fano's inequality is often used to provide a lower bound on the potential performance of a general estimator by the ostensibly simpler quantity $H(X|Z)$, as determining an optimal value for $\mathbb{P}(e)$ is often intractable.  A virtue of Theorem \ref{thm: Fano Sharpening} is that it provides upper bounds on $H(X|Z)$, in terms of tail bounds of a single log-concave variable $|W|$.  Thus, Theorem \ref{thm: Fano Sharpening} asserts that for a large class of models, $H(X|Z)$ can be controlled by a single easily computable quantity, which in the case that $M=1$, decays sub-exponentially in $\lambda$ to $0$.  However, the example delineated below, demonstrates that even in simple cases where an optimal estimator of $X$ admits explicit computation, the bounds of Theorem \ref{thm: Fano Sharpening} may outperform the best possible bounds based on Fano's inequality.

Suppose that $X$ is uniformly distributed on $\{1,2, \dots, N\}$ and that $W$ is an independent, symmetric log-concave variable with density $\varphi$, and $Z = X+W$, then $Z$ has density $f(z) = \displaystyle \sum_{i=1}^N \frac{f_i(z)}{N}$ with $f_i(z) = \varphi(z -i)$.  The optimal (Bayes) estimator of $X$ is given by $\Theta(z) = argmax_{i} \{f_i(z) : i \in \{1,2,\dots,N\}\}$, which by the assumption of symmetric log-concavity can be expressed explicitly as:
\begin{align}
    \Theta(z) = \mathbbm{1}_{(-\infty, \frac 3 2)} + \sum_{i=2}^{N-1} i \mathbbm{1}_{( i - \frac 1 2, i + \frac 1 2)} + N \mathbbm{1}_{(N-\frac 1 2, \infty)}.
\end{align}
Thus, $\mathbb{P}(\Theta \neq X)$ can be written explicitly as well.  Indeed,
\begin{align}
    \mathbb{P}( X = \Theta(Z))
        &=
            \sum_i \mathbb{P}(X = i, i = \Theta(i + W))
                \\
        &=
            \frac 1 N \mathbb{P} \left(W \leq \frac 1 2 \right) + \frac 1 N \mathbb{P}\left(W \geq - \frac 1 2 \right)+ \sum_{i=2}^{N-1} \frac{ \mathbb{P}\left( W \in (- \frac 1 2, \frac 1 2) \right)}{N} 
                \\
        &=
            \mathbb{P}\left( W \in ( - \frac 1 2, \frac 1 2) \right) + \frac 2 N \mathbb{P}( W \geq \frac 1 2 ).
\end{align}
Thus writing $P(e) = \mathbb{P}( X \neq \Theta)$, we have
\begin{align}
    P(e) &= 2 \mathbb{P}( W \geq \frac 1 2) \left( 1 - \frac 1 N \right)
        \\
        &=
            \mathbb{P}( |W| \geq 1/2) \left( 1 - \frac 1 N \right)
                \\
        &=
            \mathbb{P}( |\mathcal{Z}| \geq 1/2 \sigma) \left( 1 - \frac 1 N \right),
\end{align}
where $\mathcal{Z}$ is a standard normal variable.
Thus the optimal bounds achievable through Fano's inequality are described by,
\begin{align}
    H(X|Z) 
        \leq
            H(e) + P(e) \log(N-1)
\end{align}
with $P(e) = \mathbb{P}( |\mathcal{Z}| \geq 1/2 \sigma) \left( 1 - \frac 1 N \right)$.  Note that with $N \to \infty$, the bounds attainable through Fano's inequality become meaningless since $\lim_{N \to \infty} H(e) + P(e) \log(N-1) = \infty$ independent of $\sigma$.  
In contrast, since $Z = X+W$, Corollary \ref{cor: AGWN in 1 d bounds} gives the following bound:
\begin{align} \label{eq: 1d Gaussian bound with fano}
    H(X|Z) \leq   \log \left[3 e^{(\lambda/\sigma)^2 + 3}     \left( 1 +   \sqrt{\frac{9\pi}{2}} \frac{\sigma}{ \lambda }  \right) \right] \mathbb{P} \left(|\mathcal{Z}| > \lambda/\sigma \right),
\end{align}
independent of $N$.

\subsection{Channel Capacity}
In the case of a channel that admits discrete inputs (and possibly continuous inputs as well) with output density $f_i(z) = p(z|i)$ when conditioned on an input $i$. Suppose the  input $X$ takes value $i$ with probability $p_i$ then the output $Z$ distribution will have a density function $\sum_i p_i f_i. $ Thus 
\begin{align}
    I(Z;X) 
        &= 
            h(Z) - h(Z|X)
            \\
        &=
            H(X) - H(X|Z)
                \\
        &=
            h(\sum_i p_i f_i) - \sum_i p_i h(f_i)
                .
\end{align}
Thus, any choice of input $X$ gives a lower bound on the capacity of the channel.  In the context of additive white Gaussian noise channel \cite{ozarow1990capacity} gave rigorous bounds to the findings of \cite{ungerboeck1982channel}, that finite input can nearly achieve capacity. \\

\begin{thm}[Ozarow-Wyner \cite{ozarow1990capacity}] \label{thm: OzaWy}
   Suppose $X$ is uniformly distributed on $N$ evenly spaced points, $\{2 \lambda, 4 \lambda, \dots, 2 N \lambda \}$ and  its variance $\mathbb{E}X^2 - \mathbb{E}^2X = \lambda^2 \left( \frac{N^2 -1}{3} \right)$. If $Z = X+W$, where $W$ is Gaussian with variance one   and independent of $X$. Then 
    \begin{enumerate}
        \item 
    \begin{align} \label{eq: OzWy a}
        I(Z;X) \geq ( 1- (\pi K)^{-1/2} e^{-K}) H(X) - h( (\pi K)^{-1/2} e^{-K})
    \end{align}
    where
    \begin{align} 
        K 
            &= 
                \frac{ 3} { 2 \alpha^2} ( 1- 2^{-2 C})
                    \\
        \alpha 
            &=
                N 2^{-C}
                    \\
        C 
            &=
            \frac 1 2 \log \left( 1+ \lambda^2 \frac{N^2-1}{3} \right).
    \end{align}
    \item \label{eq: Ozarow Wyner bound that is generalized}
    \begin{align}
        I(Z;X)
            \geq
                C - \frac 1 2 \log \frac {\pi e} 6 - \frac 1 2 \log \frac{1 + \alpha^2}{\alpha^2}.
    \end{align}
    \end{enumerate}
\end{thm}

In the notation of this paper $K = \frac{\lambda^2}{2} \left( 1- \frac{1}{N^2} \right)$.  Defining,
\begin{align}
    p_o \coloneqq \frac{e^{- \frac{\lambda^2}{2} \left( 1- \frac{1}{N^2} \right)}}{\sqrt{ \frac{\pi \lambda^2}{2} \left( 1- \frac{1}{N^2} \right)}},
\end{align}
we can re-write \eqref{eq: OzWy a} as
\begin{align} \label{eq: simple ozwy a}
   I(X|Z) \geq H(X) -\left(p_o H(X) + h(p_o) \right).
\end{align}
Note that $N 2^{-C} = \frac{N}{\sqrt{1 + \sigma^2}} = \sqrt{ \frac{ 1 + 3 (\sigma/\lambda)^2}{1 + \sigma^2}}$, so that $\alpha \approx \sqrt{3}/ \lambda$ for $\sigma$ large.  Thus \eqref{eq: OzWy a} gives a bound with sub-Gaussian-like convergence in $\lambda$ to $H(X)$ for fixed $N$, but gives worse than trivial bounds for fixed $\lambda$ and $N \to \infty$.  In contrast, by Corollary~\ref{cor: AGWN in 1 d bounds} gives
\begin{align} \label{eq: capacity bound from our theorem to compare to ozawyy}
    I(X|Z)
        \geq
            H(X) - \log \left[3 e^{\lambda^2 + 3}     \left( 1 +   \sqrt{\frac{9\pi}{2}} \frac{\sigma}{ \lambda }  \right) \right] \mathbb{P} \left(|\mathcal{Z}| > \lambda \right).
\end{align}
Comparing the bound on the gap between $I(X|Z)$ and $H(X)$ provided by \eqref{eq: simple ozwy a} and Corollary \ref{cor: AGWN in 1 d bounds}, we see that Corollary \ref{cor: AGWN in 1 d bounds} outperforms Theorem \ref{thm: OzaWy} for large $\lambda$.  Indeed, one can easily find an explicit rational function $q$ such that
\begin{align}
    \frac {
    p_o H(X) + h(p_o)
    }{
    \log \left[3 e^{\lambda^2 + 3}     \left( 1 +   \sqrt{\frac{9\pi}{2}} \frac{\sigma}{ \lambda }  \right) \right] \mathbb{P} \left(|\mathcal{Z}| > \lambda \right) }
        \geq q(\lambda) e^{\frac{\lambda^2}{2N}}.
\end{align}
Additionally, \eqref{eq: capacity bound from our theorem to compare to ozawyy} gives a universal bound, independent of $N.$

These results have been of recent interest, see for example \cite{dytso2016interference, dytso2017generalized}, where the results improving and generalizing \eqref{eq: Ozarow Wyner bound that is generalized} have been studied in a form
\begin{align}
    H(X) - gap^* \leq I(Z;X) \leq H(X),
\end{align}
with an emphasis on achieving $gap^*$ bounds that are independent of $N$, and viable for more general noise models.
The significance of the results of Theorem \ref{thm: Fano Sharpening} in this context is that the $gap^*$ bounds provided converge exponentially fast to zero in $\lambda$, independent of $H(X)$, while for example in  \cite{dytso2016interference}, the $gap^*$ satisfies
\begin{align}
    gap^* \geq \frac 1 2 \log \frac{2 \pi e}{12}.
\end{align}
Additionally, the tools developed can be extended to perturbations of the $Y$ and signal dependent noise through Theorem \ref{thm: Fano Sharpening}.

A related investigation of recent interest is the relationship between finite input approximations of capacity achieving distributions, particularly the number of ``constellations'' needed to approach capacity.  For example \cite{wu2010functional, wu2010impact} the rate of convergence in $n$ of the capacity of an $n$ input power constrained additive white Gaussian noise channel to the usual additive white noise Gaussian channel is obtained.  In many practical situations, although a Gaussian input is capacity achieving, discrete inputs are used.  We direct the reader to \cite{wu2018survey} for background on the role this practice plays in Multiple Input- Multiple Output channels pivotal in the development of 5G technology.

Additionally in the amplitude constrained discrete time additive white Gaussian noise channel, the capacity achieving distribution is itself discrete \cite{smith1971information}.  In fact, many important channels achieve capacity for discrete distributions, see for example \cite{chan2005capacity,varshney2008transporting,abou2001capacity,shamai1995capacity, HGKMM18:isit}.   Thus in the case that the noise model is independent of input, the capacity achieving output will be a mixture distribution, and the capacity of the channel is  given by calculating the entropy of said mixture. 

Theorem \ref{thm: Fano Sharpening} shows that for sparse input, relative to the strength of the noise, the mutual information of the input and output distributions is sub-exponentially close to the entropy of the the input in the case of log-concave noise, and sub-Gaussian from the the entropy of the input in the case of strongly log-concave noise,  which includes Gaussian noise as a special case.  In contrast Theorem \ref{thm: discrete continuous entropy inequality sharpening} gives a reverse inequality, demonstrating that when the mixture distributions are close to one another in the sense that their total variation distance from the mixture ``with themselves removed'' is small, then the mutual information is quantifiably lessened.

\subsection{Energetics of Non-equillibrium thermodynamics}
For a process $x_t$ that satisfies an overdamped Langevin stochastic differential equation, $
d x_t = -  \frac{\nabla U(x_t,t)}{\gamma}dt + \sqrt{2D}d \zeta_t.
$ with time varying potential $U: \mathbb{R}^d \times \mathbb{R} \to \mathbb{R}$ and $\zeta_t$ a Brownian motion with $D = k_B T/\gamma$ where $\gamma$ is the viscosity constant, $k_B$ is Boltzman's constant and $T$ is temperature, one can define natural thermodynamic quantities, in particular trajectory dependent notions of work done $\mathcal{W}$ on the system (see \cite{jarzynski1997equilibrium}) and heat dissipated $\mathcal{Q}$, respectively,
\begin{align}
    \mathcal{W} \coloneqq \int_0^{t_f} \partial_t U(x_t,t) dt 
\end{align}
and
\begin{align}
   \mathcal{Q} \coloneqq -\int_0^{t_f} \nabla_x U(x_t,t) \circ dx_t,
\end{align}
where the above is a Stratonovich stochastic integral.  Recall that Stratonovich integrals satisfy a chain rule $d U(x_t,t) = \nabla_x U (x_t, t) \circ dx_t + \frac{\partial U}{\partial t} (x_t,t) dt$ so that we immediately have a first law of thermodynamics  
\begin{align}
    \Delta U 
        &\coloneqq 
            U(t_f, x(t_f)) - U(0,x(0)) 
                \\
        &= 
            \int_0^{t_f} \partial_t U(x_t,t) dt + \int_0^{t_f} \nabla_x U(x_t,t) \circ dx_t
                \\
        &=
            \mathcal{W} - \mathcal{Q}.
\end{align}  Further, if $\rho_{t}$ denotes the distribution of $x_t$ at time $t$, satisfying the Fokker-Planck equation then it can be shown  \cite{aurell2012refined} (see also \cite{chetrite2008fluctuation,maes2008steady,seifert2005entropy}),
\begin{align}
\mathbb{E} \mathcal{Q} =  k_B T \left( h(\rho_0) - h(\rho_{t_f}) \right) + \int_0^{t_f} \mathbb{E} |v(t,x_t)|^2   dt,
\end{align}
where $v$ is mean local velocity (see \cite{aurell2012refined} or as the current velocity in \cite{nelson1967dynamical}).                                
In the quasistatic limit where the non-negative term $\int_0^{t_f} \langle |v(t,x_t)|^2 \rangle  dt$ goes to $0$, one has a fundamental lower bound on the efficiency of a process's evolution, the average heat dissipated in a transfer from configuration $\rho_0$ to $\rho_{t_f}$ is bounded below by the change in entropy.
\begin{align} \label{eq: second law}
    \mathbb{E} \mathcal{Q} = k_B T \left( h(\rho_0) - h(\rho_{t_f}) \right).
\end{align}
A celebrated example of this inequality is Landauer's principle \cite{landauer1961irreversibility}, which proposes fundamental thermodynamic limits to the efficiency of a computer utilizing logically irreversable computations (see also \cite{bennett1973logical}).  More explicitly \eqref{eq: second law} suggests that the average heat dissipated in the erasure of a bit, that is, the act of transforming a random bit to a deterministic $0$ is at least $k_B T \log 2$.  This can be reasoned to in the above, by presuming the entropy of a random bit should satisfy $h(\rho_0) = \log 2$ and that the reset bit should satisfy $h(\rho_{t_f}) = 0$.

In the context of nanoscale investigations, (like protein pulling or the intracellular transport of cargo by molecular motors) it is often the case that phenomena take one of finitely many configurations with an empirically derived probability.  However at this scale, thermal fluctuations can make discrete modeling of the phenomena unreasonable, and hence the distributions $\rho_0$ and $\rho_{t_f}$ in such problems are more accurately modeled as a discrete distribution disrupted by thermal noise, and are thus, mixture distributions.  Consequently bounds on the entropy of mixture distributions translate directly to bounds on the energetics of nanoscale phenomena \cite{talukdar2018analysis,melbourne2018realizing}.  For example, in the context of Landauer's bound, the distribution of the position of a physical bit is typically modeled by a Gaussian bistable well, explicitly by the density
\begin{align} \label{eq: Landauer bit set up}
    f_p(z) = p e^{-(x-a)^2/2\sigma^2}/\sqrt{2\pi \sigma^2} + (1-p) e^{-(x+a)^2/2\sigma^2}/\sqrt{2\pi \sigma^2}.
\end{align}
The variable $p$ connotes the probability that the bit takes the value $1$, and $(1-p)$ the probability the bit takes the value $0$.  This can be modeled by $X_p$ a Bernoulli variable taking values $\pm a$ and $Z_p = X_p + \sigma \mathcal{Z}$ where $\mathcal{Z}$ is a standard normal, so that $Z_p$ has distribution $f_p$.
\begin{coro}
    The average heat dissipated $\mathcal{Q}_0$ in an optimal erasure protocol, resetting a random bit to zero in the framework of \eqref{eq: Landauer bit set up} can be bounded above and below,
    \begin{align} \label{eq: Landauer upper and lower bounds}
     \tilde{C}_L \hspace{1mm} \mathbb{P}( |\mathcal{Z} | > a /\sigma) \leq \mathbb{E} \mathcal{Q}_0  - k_B T \log 2 \leq  \tilde{C}_U \hspace{1mm} \mathbb{P}( |\mathcal{Z} | > a /\sigma)
    \end{align}
    where 
    $
        \tilde{C}_L 
            = 
                - k_B T \left( \log \left[3 e^{(a/\sigma)^2 + 3}     \left( 1 +   \sqrt{\frac{9\pi}{2}} \frac{\sigma}{ a }  \right) \right]  \right)
    $ and
    $
        \tilde{C}_U
            =
                k_B T \left( \log \left[\frac 3 2 e^{(a/\sigma)^2 + 3}     \left( 1 +   \sqrt{\frac{9\pi}{2}} \frac{\sigma}{ a }  \right) \right]  \right)$.\\
                
    More generally, in the case that the erasure is imperfect, so that the probability of failure is non-negligible we have the following bound,
    \begin{align} \label{eq: Generalized landauer upper and lower}
        C_L \hspace{1mm} \mathbb{P}( |\mathcal{Z} | > a /\sigma) \leq \mathbb{E} \mathcal{Q}_0  - k_B T \left( H(p_0) - H(p_1) \right) \leq  C_U \hspace{1mm} \mathbb{P}( |\mathcal{Z} | > a /\sigma)
    \end{align}
    where
  \begin{align} 
        C_L 
            &= 
                - k_B T \left( \log \left[3 e^{(a/\sigma)^2 + 3}     \left( 1 +   \sqrt{\frac{9\pi}{2}} \frac{\sigma}{ a }  \right) \right] - H(X_{p_1}) \right)
                    \\
        C_U
            &=
                k_B T \left( \log \left[3 e^{(a/\sigma)^2 + 3}     \left( 1 +   \sqrt{\frac{9\pi}{2}} \frac{\sigma}{ a }  \right) \right] - H(X_{p_0}) \right).
    \end{align}

\end{coro}

\begin{proof}
First let us note that we understand a random bit to be the case that $p_0 = \frac 1 2$, while an erasured bit is to be understood as a deterministic $X$ with $p_1 = 0$.  Thus, \eqref{eq: Landauer upper and lower bounds} follows immediately from \eqref{eq: Generalized landauer upper and lower}
If we let $\rho_0 = f_{p_0}$ and $\rho_{t_f} = f_{p_1}$, and let $Z_{p_i} = X_p + \sigma \mathcal{Z}$ denote a variable then \eqref{eq: second law} gives
\begin{align}
    \mathbb{E}{Q_0} 
        &=
            k_B T \left( h(f_{p_0}) - h(f_{p_1}) \right)
                \\
        &= \label{eq: heat loss as mutual info difference}
            k_B T \left( I(Z_{p_0}; X_{p_0}) - I(Z_{p_1}; X_{p_1}) \right),
\end{align}
where the second equality follows from the fact that both $Z_{p_i}$ variables are conditionally Gaussian of the same variance.  Using the corollary of Theorem \ref{thm: Fano Sharpening simple}  obtained in \eqref{eq: 1d Gaussian bound with fano},
\begin{align}
    I(Z_{p_i}; X_{p_i}) \geq H(X_{p_i}) - \log \left[3 e^{(a/\sigma)^2 + 3}     \left( 1 +   \sqrt{\frac{9\pi}{2}} \frac{\sigma}{ a }  \right) \right] \mathbb{P} \left(|\mathcal{Z}| > a/\sigma \right),
\end{align}
while Theorem \ref{thm: discrete continuous entropy inequality sharpening} applied as in \eqref{eq: first mutual information bounds} as
\begin{align}
    I(Z_{p_i}; X_{p_i}) 
        &\leq 
            \mathcal{T}_{f} H(X_{p_i}) 
            \\
        &= H(X_{p_i}) - \mathbb{P}(|\mathcal{Z}| > a /\sigma) H(X_{p_i}),
\end{align}
since $T_{f} = 1 - \mathbb{P}(|\mathcal{Z}| > a /\sigma)$.
Combining these results gives
\begin{align}
    I(Z_{p_0}; X_{p_0}) - I(Z_{p_1}; X_{p_1}) &\leq H(X_{p_0}) - H(X_{p_1}) + \mathbb{P}(|\mathcal{Z}| > a /\sigma) \left( \log \left[3 e^{(a/\sigma)^2 + 3}     \left( 1 +   \sqrt{\frac{9\pi}{2}} \frac{\sigma}{ a }  \right) \right] - H(X_{p_0}) \right)
        \\
    I(Z_{p_0}; X_{p_0}) - I(Z_{p_1}; X_{p_1}) &\geq H(X_{p_0}) - H(X_{p_1}) - \mathbb{P}(|\mathcal{Z}| > a /\sigma)  \left( \log \left[3 e^{(a/\sigma)^2 + 3}     \left( 1 +   \sqrt{\frac{9\pi}{2}} \frac{\sigma}{ a }  \right) \right] - H(X_{p_1}) \right).
\end{align}
Inserting these equations into \eqref{eq: heat loss as mutual info difference} completes the proof.
\end{proof}

\subsection{Functional inequalities}
Mixture distributions arise naturally in mathematical contexts as well.  For example in \cite{bobkovmarsiglietti2019}
Bobkov and Marsiglietti found interesting application of $h(X+Z) \leq H(X) + h(Z)$ for $X$ discrete and $Z$ independent and continuous in the investigation of entropic Central Limit Theorem for discrete random variables under smoothing.

 In the study of stability in the Gaussian log-Sobolev inequalities, Eldan, Lehec, and Shenfeld \cite{eldan2019stability}, it is proven as Proposition~5 that the deficit in the Gaussian log-Sobolev inequality, defined as
\begin{align}
    \delta(\mu) = \frac{I(\mu || \gamma)} 2 - D( \mu || \gamma) 
\end{align}
for a measure $\mu$ and $\gamma$ the standard $d$-dimensional Gaussian measure, and $I$ is relative Fisher information,
\begin{align}
    I(\mu || \gamma) \coloneqq \int_{\mathbb{R}^d} \log \left( \frac{ d\mu}{d\gamma} \right) d\gamma,
\end{align}
is small for Gaussian mixtures.  More explicitly for $p_i$ non-negative numbers summing to $1$,
\begin{align}
    \delta \left(\sum_i p_i \gamma_i \right) \leq H(p).
\end{align}
In the language of Theorem \ref{thm: discrete continuous entropy inequality sharpening} a sharper bound can be achieved.
\begin{coro}
    When $\gamma_i$ are translates of the standard Gaussian measure then
    \begin{align}
        \delta \left( \sum_i p_i \gamma_i\right) \leq \mathcal{T} H(p),
    \end{align}
    where $\mathcal{T}$ is defined as in Theorem \ref{thm: discrete continuous entropy inequality sharpening}.
\end{coro}

\begin{proof}
    By the convexity of the relative Fisher information, the equality $D( \sum_i p_i \gamma_i || \gamma) = \sum_i p_i D(\gamma_i || \gamma) + h(\sum_i p_i \gamma_i) - \sum_i p_i h(\gamma_i)$, $\frac{I(\gamma_i || \gamma)}{2} - D(\gamma_i || \gamma) =0 $, and the application of Theorem \ref{thm: discrete continuous entropy inequality sharpening} we have
        \begin{align}
            \delta\left( \sum_i p_i \gamma_i \right) 
                &=
                    \frac{I( \sum_i p_i \gamma_i || \gamma)}{2} - D( \sum_i p_i \gamma_i || \gamma)
                        \\
                &\leq
                    \sum_i p_i \left(\frac{I(\gamma_i || \gamma)}{2} - D(\gamma_i || \gamma) \right) + \left( h(\sum_i p_i \gamma_i) - \sum_i p_i h(\gamma_i) \right)
                        \\
                &=
                    h(\sum_i p_i \gamma_i) - \sum_i p_i h(\gamma_i)
                        \\
                &\leq 
                    \mathcal{T} H(p).
        \end{align}
\end{proof}

\section{Conclusions}
In this article, the entropy of mixture distributions is estimated by providing tight upper and lower bounds. The efficacy of the bounds is demonstrated, for example, by demonstrating that    existing bounds on the conditional entropy, $H(X|Z)$ of a random variable, $X$ taking values in a countable set, $\mathcal{X}$ conditioned on a continuous random variable, $Z$, become meaningless as the cardinality of the set $\mathcal{X}$ increases while the bounds obtained here remain relevant. Significantly enhanced upper bounds on mutual information of channels that admit discrete input with continuous output are obtained based on the bounds on the entropy of mixture distributions. The technical methodology developed is of interest in its own right whereby connections to existing results either can be derived as corollaries of more general theorems in the article or are improved upon by the results in the article. These include the reverse Pinsker inequality, and bounds on Jensen-Shannon divergence, and bounds that are obtainable via Fano's inequality.

\section{Acknowledgement}
The authors acknowledge the support of the National Science Foundation for funding the research under 
Grant No. 1462862 (CMMI), 1544721 (CNS) and  1248100 (DMS).


\appendix

\section{Proof of Skew Relative Information Properties}

\begin{proof}[Proof of Proposition \ref{prop: divergence skew properties}]
There is nothing to prove in \eqref{item: skew divergence to KL divergence}, this is exactly the definition of the usual relative entropy from $\mu$ to $t \mu + (1-t) \nu$. For \eqref{item: D one is zero}, by \eqref{item: skew divergence to KL divergence} $S_t(\mu || \nu) =0$ iff $D(\mu || t \mu + (1-t) \nu) =0$ which is true iff $\mu = t \mu + (1-t) \nu$ which happens iff $t=1$ or $\mu = \nu$.
To prove \eqref{item: skew divergence bounded by log t}, observe that for a Borel set $A$
\[
    \mu(A) \leq \frac 1 t (t \mu + (1-t) \nu) (A).
\]
This gives the following inequality, from which absolute continuity, and the existence of $\frac{d \mu}{d (t \mu + (1-t) \nu)}$ follow immediately,
\begin{align} \label{eq: radon nikodym derivative bound}
    \frac{d \mu}{d(t \mu + (1-t) \nu)} \leq \frac 1 t.
\end{align}
Integrating \eqref{eq: radon nikodym derivative bound} gives
\begin{align}
    S_t(\mu || \nu) \leq - \log t. 
\end{align}
To prove \eqref{item: function of t}, notice that for fixed $\mu$ and $\nu$, the map $\Phi_t=t \mu + (1-t) \nu$ is affine, and since the relative entropy is jointly convex \cite{CT91:book}, convexity in $t$ follows from the computation below.  
\begin{align*}
    S_{(1-\lambda)t_1 + \lambda t_2}(\mu || \nu)
            &=
            D(\mu || \Phi_{(1-\lambda)t_1 + \lambda t_2} )
                \\
        &=
            D(\mu || (1-\lambda) \Phi_{t_1} + \lambda \Phi_{ t_2})
                \\
        &\leq
            (1-\lambda) D(\mu || \Phi_{t_1}) + \lambda D( \mu || \Phi_{t_2})
                \\
        &=
            (1-\lambda) S_{t_1} ( \mu || \nu) + \lambda S_{t_2}( \mu || \nu).
\end{align*}

Since $t \mapsto S_t(\mu || \nu)$ is a non-negative convex function on $(0,1]$ with $S_1(\mu || \nu) = 0$ it is necessarily non-increasing. When $\mu \neq \nu$, $\mu \neq t\mu + (1-t) \nu$ so that $S_t(\mu || \nu) >0$ for $t < 1$, so that as a function of $t$ the skew divergence is strictly decreasing.
To prove that $S_t$ is an $f$-divergence recall Definition \ref{def: f-divergence}.  It is straight forward that $S_t$ can be expressed in form \eqref{eq: f -divergence definition} with $f(x) = x \log (x/(tx + (1-t))$. Convexity of $f$ follows from the second derivative computation, 
\[
    f''(x) = \frac{ (t-1)^2}{x ( tx + (1-t))^2} > 0.
\]
Since $f(1) = 0$ the proof is complete.
\end{proof}

\section{Gaussian Bounds}
In this section we consider $W \sim \varphi_\sigma$ with $\varphi_\sigma(w) = e^{-|w|^2/2 \sigma} /(2 \pi \sigma^2)^{\frac d 2}$, and use $\varphi$ to denote $\varphi_1$ and use $\mathcal{Z}$ in place of $W$ in this case.
\begin{prop} \label{prop: entropy tail bounds}
For $d \geq 2$
\begin{align}  \label{eq: entropy tail bounds}
 - \int_{B_\lambda^c}  \varphi_\sigma(w) \log \varphi_\sigma(w) dw \leq \left( \frac d 2 \log 2 \pi e^2 \sigma^2 + \frac{\lambda^2}{\sigma^2} \right) \hspace{.5 mm} \mathbb{P}(|W| > \lambda)
 \end{align}
\end{prop}

\begin{proof}
    We first show that the result for general $\sigma$ follows from the case that $\sigma = 1$.  Indeed, assuming \eqref{eq: entropy tail bounds}, the substitution $u = w/\sigma$ gives
    \begin{align}
        \int_{B_\lambda^c} \varphi_\sigma(w) \log \varphi_\sigma(w) dw
            &=
                \mathbb{P}(|\mathcal{Z}| > \lambda /\sigma ) d \log \sigma  - \int_{B_{\lambda/\sigma}^c} \varphi(u) \log \varphi(u) du
                    \\
            &\leq
                \left[ \frac d 2 \log 2 \pi e^2 \sigma^2 + \frac{\lambda^2}{\sigma^2} \right] \mathbb{P}(|\mathcal{Z}| > \lambda/\sigma),
    \end{align}
    where we have applied \eqref{eq: entropy tail bounds} to achieve the inequality.  Since $W$ has the same distribution as $\sigma \mathcal{Z} $, the reduction holds.
    By direct computation,
    \begin{align}
        - \int_{B_\lambda^c}  \varphi(w) \log \varphi(w) dw
            &=
                \int_{B_\lambda^c} \varphi(w) \frac d 2 \log 2 \pi  dw + \int_{B_\lambda^c} \frac{|w|^2}{2}  \varphi(w) 
                    \\
            &=
                \mathbb{P}(|\mathcal{Z}| > \lambda) \left( \frac d 2 \log 2 \pi  + \frac{ (2 \pi)^{-d/2} \omega_d \int_\lambda^\infty r^{d+1} e^{-r^2/2} dr}{ (2 \pi)^{-d/2} \omega_d \int_\lambda^\infty r^{d-1} e^{-r^2/2} dr} \right)
                    \\
            &=
                \mathbb{P}(|\mathcal{Z}| > \lambda) \left( \frac d 2 \log 2 \pi + \frac{ \lambda^d e^{-\lambda^2/2} + d \int_\lambda^\infty r^{d-1} e^{-r^2/2}}{\int_\lambda^\infty r^{d-1} e^{-r^2/2} dr} \right)
                    \\
            &=
                \mathbb{P}(|\mathcal{Z}| > \lambda) \left( \frac d 2 \log 2 \pi e^2 + \frac{ \lambda^d e^{-\lambda^2/2} }{ \int_\lambda^\infty r^{d-1} e^{-r^2/2} dr} \right).
    \end{align}
    Using $r^{d-1} \geq r \lambda^{d-2}$ for $r \geq \lambda$ when $d \geq 2$, 
    \begin{align}
        \int_\lambda^\infty r^{d-1} e^{-r^2/2} dr \geq \lambda^{d-2} \int_\lambda^\infty r e^{-r^2/2} dr = \lambda^{d-2} e^{-r^2/2}.
    \end{align}
    Thus,
    \begin{align}
                - \int_{B_\lambda^c}  \varphi(w) \log \varphi(w) dw \leq \mathbb{P}(|\mathcal{Z}| > \lambda) \left( \frac d 2 \log 2 \pi e^2 + \lambda^2 \right)
    \end{align}
\end{proof}

\begin{prop} \label{prop: appendix regular norm bound}
For $d\geq 2$,
\begin{align} \label{eq: appendix regular norm bound}
    \int_{B_\lambda^c} \varphi(w) |w| dw \leq ( \lambda + d \sigma ) \mathbb{P}(|W| > \lambda)
\end{align}
\end{prop}

\begin{proof}
Again we reduce to the case that $\sigma = 1$. Substituting $u = w/\sigma$ gives
\begin{align}
    \int_{B_\lambda^c} \varphi_\sigma(w) |w| dw 
        &=
            \sigma \int_{B_{\lambda/\sigma}^c} \varphi(u) |u| du
                \\
        &\leq
            \left( \frac \lambda \sigma + d \right) \sigma \mathbb{P}(|\mathcal{Z}| > \lambda/\sigma)
                \\
        &=
            ( \lambda + d \sigma) \mathbb{P}(|W| > \lambda),
\end{align}
where we have used \eqref{eq: appendix regular norm bound} for the inequality and $\sigma \mathcal{Z}$ being equidistributed with $W$ for the last equality.  We now proceed in the reduced case.
By change of coordinates and integration by parts,
\begin{align}
    \frac{ \int_{B_\lambda^c} \varphi(w) |w| dw}{ \mathbb{P}(|\mathcal{Z}| > \lambda)}
        &=
            \frac{\int_\lambda^\infty r^{d} e^{-r^2/2} dr }{\int_\lambda^\infty r^{d-1} e^{-r^2/2} dr }
                \\
        &=
            \frac{ - r^{d-1} e^{-\lambda^2/2} \big|_\lambda^\infty  + d\int_\lambda^\infty r^{d-1} e^{-r^2/2} dr }{\int_\lambda^\infty r^{d-1} e^{-r^2/2} dr }
                \\
        &=
            \frac{ \lambda^{d-1} e^{-\lambda^2/2}}{\int_\lambda^\infty r^{d-1} e^{-r^2/2} dr  } + d
\end{align}
Using $r \geq \lambda$, for $d\geq 2$,
\begin{align}
    \int_\lambda^\infty r^{d-1} e^{-r^2/2} dr \geq \lambda^{d-2} \int_\lambda^\infty r e^{-r^2/2} dr = \lambda^{d-2} e^{-\lambda^2/2}.
\end{align}
Thus,
\begin{align}
    \frac{ \int_{B_\lambda^c} \varphi(w) |w| dw}{ \mathbb{P}(|\mathcal{Z}| > \lambda)} \leq \lambda + d.
\end{align}
\end{proof}

\begin{prop} \label{prop: Gaussian density versus tail}
    When $d=1$, so that $W \sim \varphi_\sigma(w) = e^{-x^2/2 \sigma}/\sqrt{2 \pi \sigma^2}$, we have the following bounds for $\lambda > 0$,
    \begin{align}
        \sigma^2 \varphi_\sigma(\lambda) = \int_\lambda^\infty w \varphi_\sigma(w) dw 
            &\leq 
                \int_\lambda^\infty \varphi_\sigma(w) dw  \left( \lambda + \frac {\sigma^2} \lambda \right) \label{eq: Tail lower bound d=1}
            \\
        - \int_{\lambda}^\infty \varphi_\sigma(w) \log \varphi_\sigma (w) dw
            &\leq
                \left( \lambda^2/\sigma^2 +2 + \log (\sqrt{2 \pi} \sigma) \right) \int_\lambda^\infty \varphi_\sigma (w) dw   \label{eq: entropy tail d=1}
    \end{align}
\end{prop}
\begin{proof}
The inequality \eqref{eq: Tail lower bound d=1} is standard.  The inequality can be reduced to the $\sigma =1$ by applying \eqref{eq: Tail lower bound d=1} after change of variables $u= w/\sigma$. The proof then follows from the $\sigma = 1$ case. Recall $\varphi'(w) = w \varphi(w)$ and observe that the function 
\begin{align}
    g(\lambda) = \int_\lambda^\infty \varphi(w) dw - \frac{\lambda^2}{\lambda^2 +1} \varphi(\lambda)
\end{align}
satisfies $g(0) > 0$, $\lim_{\lambda \to \infty} g(\lambda) = 0$, and has derivative 
\begin{align}
    \frac{-2 \varphi(\lambda)}{(\lambda^2 + 1)^2} < 0,
\end{align}
so that $g(\lambda) >0$ which is equivalent to \eqref{eq: Tail lower bound d=1}.  To prove
\eqref{eq: entropy tail d=1}, we again reduce to the $\sigma = 1$ case by the substitution $u = w/\sigma$.  Then compute directly using integration by parts,
\begin{align}
    - \int_{\lambda}^\infty \varphi(w) \log \varphi(w) dw
        &=
             \int_\lambda^\infty \log \sqrt{2 \pi} \varphi(w) dw + \int_\lambda^\infty w^2 \varphi(w) dw
                \\
        &=
             \log \sqrt{2 \pi} \int_\lambda^\infty \varphi(w) dw  +  \lambda \varphi(\lambda) + \int_\lambda^\infty \varphi(w) dw
                \\
        &\leq
            \left(\lambda^2 + 2 + \log \sqrt{2\pi} \right) \int_\lambda^\infty \varphi(w) dw.
\end{align}
The inequality is an application of \eqref{eq: Tail lower bound d=1}.
\end{proof}

\begin{prop} \label{prop: Appendix d=1 bounds}
When $X$ and $Z$ satisfy the conditions of section \ref{sec: Upper bounds} for the one dimensional Gaussian $W \sim \varphi_\sigma(w) = e^{-w^2/2 \sigma}/\sqrt{2 \pi \sigma^2}$,
\begin{align}
H(X|Z) \leq (M-1) \mathbb{P}(|W| \leq \tau \lambda) + J(\varphi) \mathbb{P}(|W| \geq \lambda)
\end{align}
with 
\begin{align}
   J(\varphi) =     \log \left[ e^{(\lambda/\sigma)^2 + M+2} \tau M \sigma \sqrt{2 \pi}  \left( 1 + \tau + \tau^2+ \frac{\tau^2\sigma} {\lambda^2} \right)  \left( (2 \pi \sigma^2)^{- \frac 1 2} +  \frac{3}{2 \lambda }  \right) \right]
\end{align}
\end{prop}

\begin{proof}
As in the proof of Theorem \ref{thm: Fano Sharpening},
\begin{align}
    H(X|Z) 
        =& \sum_i p_i \int \varphi(w) \log \left( 1 + \frac{ \sum_{j \neq i} p_j \varphi(T_{ji}(w)) det(T_{ji}'(w)) }{p_i \varphi(w)} \right) dw
            \\
        \leq&
           \int_{B_\lambda^c} \varphi(w) \log \left( 1 + \frac{ \sum_{j} p_j  \sum_{i \neq j} \varphi(T_{ji}(w)) det(T_{ji}'(w)) }{\varphi(w)} \right) dw 
                \\
        &+  \int_{B_\lambda} \varphi(w) \log \left( 1 + \frac{ \sum_{j} p_j  \sum_{i \neq j} \varphi(T_{ji}(w)) det(T_{ji}'(w)) }{\varphi(w)} \right) dw.
\end{align}
with
\begin{align}
     \int_{B_\lambda} \varphi(w) \log &\left( 1 + \frac{ \sum_{j} p_j  \sum_{i \neq j} \varphi(T_{ji}(w)) det(T_{ji}'(w)) }{\varphi(w)} \right) dw 
        \\
        &\leq 
        (M-1) \mathbb{P}(|W| \leq \lambda \tau) + M \mathbb{P}(|W| > \lambda). \label{eq: Gaussian outside the ball piece d=1}
\end{align}

Splitting the integral,
\begin{align}
    \int_{B_\lambda^c} &\varphi(w) \log \left( 1 + \frac{ \sum_j p_j \sum_{i\neq j} \varphi(T_{ji}(w)) det(T_{ji}'(w))}{ \varphi(w)} \right) dw
        \\
        &=
            \int_{B_\lambda^c} \varphi(w) \log \left(   \sum_j p_j \sum_{i} \varphi(T_{ji}(w)) det(T_{ji}'(w)) \right) dw - \int_{B_\lambda^c} \varphi(w) \log \varphi(w) dw.
\end{align}

Using Corollary \ref{cor: bounds for the sums},  Jensen's inequality, and \eqref{eq: Tail lower bound d=1},
\begin{align}
    \int_{B_\lambda^c} \varphi(w) \log &\left(   \sum_j p_j \sum_{i} \varphi(T_{ji}(w)) det(T_{ji}'(w)) \right) dw 
    \\
    &\leq
            \int_{B_\lambda^c} \varphi(w) \log \left[ \tau M \left( 1 + \tau + \frac{\tau^2 |w|}{\lambda} \right) \left( \|\varphi\|_\infty + \left( \frac 3 {\lambda} \right) \frac 1 2 \right) \right] 
                \\
    &\leq 
        \mathbb{P}(|W| > \lambda) \log \left[ \tau M \left( 1 + \tau + \frac{\tau^2 \int \mathbbm{1}_{\{|w| > \lambda \}}\varphi(w) |w|dw}{\mathbb{P}(|W|> \lambda) \lambda} \right)  \left( \frac{1}{\sqrt{2 \pi} \sigma} +  \frac{3}{2 \lambda} \right) \right]
            \\
    &\leq
        \mathbb{P}(|W| > \lambda) \log \left[ \tau M \left( 1 + \tau + \frac{\tau^2 (\lambda + \frac \sigma \lambda)}{ \lambda} \right)  \left( \frac 1 {\sqrt{2 \pi} \sigma} + \frac{3}{2 \lambda} \right) \right]. \label{eq: the sum piece of the complement bound d =1}
\end{align}

Then applying \eqref{eq: entropy tail d=1},
\begin{align}
   - \int_{B_\lambda^c}  \varphi(w) \log \varphi(w) dw \leq \left( (\lambda/\sigma)^2 + 2 + \log \sqrt{2 \pi}\sigma   \right) \hspace{.5 mm} \mathbb{P}(|W| > \lambda) \label{eq: entropy tail bound gaussian d =1}
\end{align}

Combining \ref{eq: Gaussian outside the ball piece d=1}, \ref{eq: the sum piece of the complement bound d =1}, and  \ref{eq: entropy tail bound gaussian d =1} we have
\begin{align}
   H(X|Z) \leq (M-1) \mathbb{P}(|W| \leq \lambda \tau)  + J(\varphi) \mathbb{P}(|W| > \lambda)
\end{align}
with 
\begin{align}
    J(\varphi) =     \log \left[ e^{(\lambda/\sigma)^2 + M+2} \tau M \sigma \sqrt{2 \pi}  \left( 1 + \tau + \frac{\tau^2 (\lambda + \frac \sigma \lambda)}{ \lambda} \right)  \left( (2 \pi \sigma^2)^{- \frac 1 2} +  \frac{3}{2 \lambda }  \right) \right]
\end{align}

\end{proof}

\end{document}